\documentclass[journal]{IEEEtran}
\usepackage{cite}

\ifCLASSINFOpdf
   \usepackage[pdftex]{graphicx}
\else
\fi
\usepackage{hyperref}
\usepackage{multirow}

\usepackage{amsmath}
\usepackage{amsthm}
\usepackage{amssymb}
\usepackage{amsbsy}
\usepackage{cleveref}
\usepackage{algorithm}
\usepackage{algpseudocode}
\DeclareMathOperator*{\argmin}{argmin}

\newcommand{\etal}{\textit{et al}}
\newcommand{\Fig}{\text{Fig. }}
\usepackage{color}

\newtheorem{lemma}{Lemma}

\begin{document}
%
\title{Learning Convolutional Sparse Coding on Complex Domain for Interferometric Phase Restoration}
%
%
%

\author{Jian~Kang,~\IEEEmembership{Member,~IEEE,}
        Danfeng~Hong,~\IEEEmembership{Member,~IEEE,}
        Jialin~Liu,
        Gerald~Baier,
        Naoto~Yokoya,~\IEEEmembership{Member,~IEEE,}
        and Beg\"um~Demir,~\IEEEmembership{Senior Member,~IEEE}

\thanks{This work was supported by the European Research Council (ERC) through the ERC-2017-STG BigEarth Project under Grant 759764, and also supported by the Japan Society for the Promotion of Science (JSPS) under Grant KAKENHI 18K18067. (\emph{Corresponding author: Danfeng Hong}).}
\thanks{J. Kang and B. Demir are with Faculty of Electrical Engineering and Computer Science, Technische Universit\"at Berlin (TU Berlin), 10587 Berlin, Germany. (e-mail: jian.kang@tu-berlin.de; demir@tu-berlin.de)}
\thanks{D. Hong is with Univ. Grenoble Alpes, CNRS, Grenoble INP, GIPSA-lab, 38000 Grenoble, France. (e-mail: hongdanfeng1989@gmail.com)}
\thanks{J. Liu is with Department of Mathematics, University of California, Los Angeles (UCLA), Los Angeles, USA (e-mail: liujl11@math.ucla.edu).}
\thanks{G. Baier and N. Yokoya are with the Geoinformatics Unit, RIKEN Center for Advanced Intelligence Project (AIP), RIKEN, 103-0027 Tokyo, Japan. (e-mail: gerald.baier@riken.jp; naoto.yokoyariken.jp)}}

%
%

\markboth{IEEE Transactions on Neural Networks and Learning Systems, 2020}%
{Shell \MakeLowercase{\textit{et al.}}: Bare Demo of IEEEtran.cls for IEEE Journals}
%



\maketitle

\begin{abstract}
\textcolor{blue}{This is the pre-acceptance version, to read the final version please go to IEEE Transactions on Neural Networks and Learning Systems on IEEE Xplore.} Interferometric phase restoration has been investigated for decades and most of the state-of-the-art methods have achieved promising performances for InSAR phase restoration. These methods generally follow the nonlocal filtering processing chain aiming at circumventing the staircase effect and preserving the details of phase variations. In this paper, we propose an alternative approach for InSAR phase restoration, \textit{i.e.} Complex Convolutional Sparse Coding (ComCSC) and its gradient regularized version. To our best knowledge, this is the first time that we solve the InSAR phase restoration problem in a deconvolutional fashion. The proposed methods can not only suppress interferometric phase noise, but also avoid the staircase effect and preserve the details. Furthermore, they provide an insight of the elementary phase components for the interferometric phases. The experimental results on synthetic and realistic high- and medium-resolution datasets from  TerraSAR-X StripMap and Sentinel-1 interferometric wide swath mode, respectively, show that our method outperforms those previous state-of-the-art methods based on nonlocal InSAR filters, particularly the state-of-the-art method: InSAR-BM3D.
The source code of this paper will be made publicly available for reproducible research inside the community.

\end{abstract}


%
\IEEEpeerreviewmaketitle

\section{Introduction}
\subsection{Interferometric Phase Restoration}
%
%
%
%
\IEEEPARstart{D}{ue} to its all-weather capability, up to decimeter spatial resolution and high sensitivity to deformation and height changes Synthetic Aperture Radar (SAR) plays an important role in remote sensing from airborne and spaceborne platforms.
By creating interferograms of SAR images acquired at different points in time or from changing platform positions, geophysical parameters, such as heights and displacement rates, can be extracted by analyzing the interferometric phase.

As a result of the coherence loss between acquisitions, the interferometric phase is corrupted by noise.
Noise removal is consequently an almost obligatory step, not only to increase the measurements' accuracy but also to ease the subsequent phase unwrapping.
One of the most straightforward methods of noise mitigation is averaging all phases inside a predefined spatial neighborhood, so-called \textit{boxcar} filtering.
Although easy to implement and fast to compute, the penalty is a degradation of the spatial resolution.
For overcoming such limitation, many advanced filters have been introduced, such as Lee's sigma filter \cite{lee1983simple} FIXME \cite{lee1998insarfilter} which utilizes statistical test for the pixel selection during the averaging, and Goldstein filter \cite{goldstein1998radar} which leverages the local power spectrum estimation of the signal for lowering the noise component.
In recent years, nonlocal-filtering based approaches, which were first applied to optical natural images \cite{buades2005non}, have received great attention from the image processing community.
Fundamentally, by averaging a group of similar pixels selected in a nonlocal manner, noise can be efficiently mitigated without degrading image details.
Such method also became a topic of extensive research in the SAR community, such as SAR amplitude imagery denoising \cite{deledalle2009iterative,parrilli2011nonlocal,cozzolino2013fast,di2016scattering}, interferometric phase denoising \cite{deledalle2011nl,deledalle2015nl,chen2013interferometric,lin2014nonlocal,sica2018insar,gao2019novel}, and PolSAR imagery restoration \cite{deledalle2015nl,chen2010nonlocal,d2013iterative}.
The nonlocal processing chain has also been extended for phase restoration of SAR stacks with the application of differential SAR interferometry \cite{sica2015nonlocal} and the preprocessing step of 3D reconstruction based on TomoSAR \cite{d2017nonlocal}.

Although nonlocal-filtering based approaches are very popular in the field of InSAR denoising, Hao \etal \cite{hongxing2015interferometric} propose a sparse coding model for approximating InSAR phase patches based on the linear combination of the learned atoms  in a dictionary. In particular, given a signal $ \mathbf{s}\in \mathbb{C}^N$ and a \textit{dictionary} matrix $ \mathbf{D}\in\mathbb{C}^{N\times M} $, $ \mathbf{s} $ can be represented by a linear combination of only few of the columns in $ \mathbf{D} $, \textit{i.e.}, $ \mathbf{s}\approx\mathbf{D}\mathbf{x} $, where $ \mathbf{x}\in\mathbb{C}^M $ is sparse. The problem of computing the sparse representation $ \mathbf{x} $ for $ \mathbf{s} $, given the dictionary $ \mathbf{D} $ is termed as \textit{sparse coding}. It can be formulated as Basis Pursuit Denoising (BPDN) problem as follows \cite{chen2001atomic}:
\begin{equation}
	\argmin_\mathbf{x}\frac{1}{2}\|\mathbf{D}\mathbf{x}-\mathbf{s}\|_2^2+\lambda\|\mathbf{x}\|_1
\label{eq:BPDN}
\end{equation}
Such sparse representation model has been a well-established tool for a very broad range of signal and image processing applications \cite{bruckstein2009sparse,hong2015novel,kang2017robust,hong2018sulora,kang2018object,hong2019augmented,huizhang2020}. Also, extensive research for SAR data modeling based on sparse representations have also been done in recent years, such as image classification \cite{zhang2015fully,yang2016riemannian,hong2017learning,zhong2017unsupervised,ren2018polsar,hang2019cascaded} and imagery denoising \cite{xu2015iterative,liu2017sar}.

Due to the computational cost, the signal $\mathbf{s}$ in Problem (\ref{eq:BPDN}) is usually a small patch rather than the entire image in practice. In order to compute a sparse representation for an entire image, those conventional methods based on solving the problem (\ref{eq:BPDN}) should be processed independently on a set of overlapping blocks covering the image. For reconstructing the whole image, the restored results of such overlapping blocks are stitched together by averaging the overlapping parts.

However, such process ignores the consistency of image pixels,
that is, any two patches should share the same pixel values on their overlapping area.
Moreover, the local structures and textures may be inevitably changed due to the application of aggregation and averaging strategies to the final value of each pixel.
Also, those strategies can induce the over-smoothing of details in images \cite{liu2016image,hong2016robust}.

Recently, to fix these issues, \textit{convolutional sparse coding}
is proposed \cite{zeiler2010deconvolutional,chalasani2013fast,bristow2013fast,heide2015fast,wohlberg2016efficient,liu2017online,liu2018first}. By replacing the dictionary $ \mathbf{D} $ with a set of convolutional filters $ \{\mathbf{d}_m \} $, the associated reconstruction of $ \mathbf{s} $ from sparse representations $ \{\mathbf{x}_m \} $ is $ \mathbf{s}\approx\sum_m\mathbf{d}_m\ast\mathbf{x}_m $, where $ \mathbf{s} $ can be an entire image rather than a small image patch and $ \ast $ denotes the convolutional operator.
Since the convolutional operator is computationally cheap in the Fourier domain \cite{wu2019orsim,wu2019fourier}, such convolutional representation
can be obtained by a global optimization in the entire image space,
that is,
Convolutional Basis Pursuit Denoising (CBPDN):
\begin{equation}
	\centering
	\argmin_{\{\mathbf{x}_m \}}\frac{1}{2}\|\sum_m\mathbf{d}_m\ast\mathbf{x}_m - \mathbf{s} \|_2^2 + \lambda\sum_m\|\mathbf{x}_m\|_1.
	\label{eq:CBPDN}
\end{equation}

Based on the success of
the convolutional sparse coding model
in natural image processing \cite{liu2016image,gu2015convolutional,zhang2017convolutional,zhang2017convolutional2}, we seek to propose the corresponding approach in the complex domain to: 1) investigate the sparse representation for InSAR phase in a convolutional manner, 2) evaluate its performance for InSAR phase restoration.

\subsection{Contributions of this paper}
The contributions of this paper can be summarized as follows:
\begin{itemize}
	\item To avoid the staircase effect and preserve the details of phase variations, we propose a complex convolutional sparse coding (ComCSC) algorithm and its gradient regularized version (ComCSC-GR) for interferometric phase restoration. To our best knowledge, this is the first time to investigate the problem of the phase restoration in a deconvolutional manner.
	\item Superior to the conventional sparse coding (SC) model in Eq. \eqref{eq:BPDN} processed on image patches, the proposed ComCSC-based methods can progressively decompose the image from local attention to global aggregation by means of the deconvolutional manner, which can provide an insight for the elementary phase components for the interferometric phases.
	\item Beyond the conventional convolutional sparse coding (CSC) model, the resulting ComCSC and its variant perform the image coding on the complex domain. In particular, we theoretically prove the feasibility of complex-valued sparse coding and provide the corresponding update rule. Additionally, the proposed ComCSC model is an extended version targeting at processing complex signals effectively, which enables the CSC to be applicable on the complex domain.
	\item The effectiveness and superiority of the proposed methods have been quantitatively demonstrated on synthetic and real datasets and compared to other state-of-the-art methods. We will open the source codes to enable reproducible research.
\end{itemize}
\subsection{Structure of this paper}
The rest of this paper is organized as follows. Section \ref{sc:ccdl} introduces the proposed convolutional dictionary learning model in complex domain. In Section \ref{sc:ccdl_GR}, we propose a complex convolutional dictionary learning model with the regularization of gradients. Simulated experiments and real case study are conducted in Section \ref{sc:experiments}. Section \ref{sc:conclustion} draws the conclusion of this paper.

\section{Complex Convolutional Dictionary Learning}\label{sc:ccdl}

Before solving CBPDN (\ref{eq:CBPDN}), we should learn a set of convolutional filters $\{\mathbf{d}_m\}_{m=1}^M$ from a batch of \textit{clean} interferograms $ \{\mathbf{s}_k\}_{k=1}^K$, which are also termed as \textit{training} interferograms. To achieve this point, we propose the Complex Convolutional Dictionary Learning (CCDL) problem:
\begin{equation}
\begin{aligned}
	&\argmin_{\{\mathbf{d}_m\},\{\mathbf{x}_{m,k} \}}\frac{1}{2}\sum_{k=1}^{K}\|\sum_{m=1}^{M} \mathbf{d}_m\ast\mathbf{x}_{m,k}-\mathbf{s}_k\|_2^2\\
	&+\lambda\sum_{m=1}^{M}\sum_{k=1}^{K}\|\mathbf{x}_{m,k}\|_1, \;
	{\rm s.t.} \; \|\mathbf{d}_m\|_2=1, \; \forall m,
\end{aligned}
\label{eq:cpx_conv_dl}
\end{equation}
where $ \mathbf{s}_k\in \mathbb{C}^N$ is one of the training interferograms with $ N $ pixels, $ K $ is the total number of training interferograms, $ \{\mathbf{x}_{m,k} \} \in \mathbb{C}^N$ and $ \{\mathbf{d}_m \} \in \mathbb{C}^L$ denote the sets of complex sparse coefficient maps and complex filters respectively, $ M $ is the number of filters, $ L $ is the size of each filter $ \mathbf{d}_m $, and the constraint indicates the normalization of learned filters. (For notation simplicity, interferograms and the coefficient maps are considered to be $ N $ dimensional vectors, and filters are $ L $ dimensional vectors.)
In this paper, given a complex-valued vector $\mathbf{x}\in\mathbb{C}^N$, the norms are defined as \[\|\mathbf{x}\|_1 = \sum_{i=1}^N |x_i|, ~~\|\mathbf{x}\|_2 = \sqrt{\sum_{i=1}^N |x_i|^2},\]where $x_i\in\mathbb{C}$ is the $i$-th element of vector $\mathbf{x}$, and $|\cdot|$ means the amplitude value of a complex number. It is worth noting that different from the conventional convolutional dictionary learning models, the sparse coefficient maps $ \{\mathbf{d}_m \} $ and the convolutional kernels are all set in complex domain. The main idea for convolutional dictionary learning is to decompose the input signal $ \mathbf{s}_k $ in a deconvolutional manner. If the signal lies in complex domain, one can also constrain the sparse coefficient maps as real numbers, and the convolutional kernels are complex, or the convolutional kernels are real and the sparse coefficient maps are complex. However, we found that both the kernels and the sparse coefficient maps are expected to be complex in order to maintain the information from the original complex signals as much as possible. More specifically, on one hand, the phase information of the signal is not lost through the convolutional decomposition. On the other hand, since a real number can be considered as a complex number with zero on its imaginary part, it can also be learned during the optimization of those learnable variables. Therefore, the proposed complex convolutional dictionary learning model can be considered as a generalization of the normal one.

An usual optimization strategy to solve problem (\ref{eq:cpx_conv_dl}) is alternately updating the sparse coefficient maps $\{\mathbf{x}_{m,k} \}$ and the dictionary $\{\mathbf{d}_m \}$.

\subsection{Sparse Coefficients Update}
We first fix the filters $ \{\mathbf{d}_m \}$ and update the sparse coefficient maps $\{\mathbf{x}_{m,k} \}$ in (\ref{eq:cpx_conv_dl}) with
\begin{equation}
	\argmin_{\{\mathbf{x}_{m,k} \}}\frac{1}{2}\sum_{k=1}^{K}\|\sum_{m=1}^{M} \mathbf{d}_m\ast\mathbf{x}_{m,k}-\mathbf{s}_k\|_2^2+\lambda\sum_{m=1}^{M}\sum_{k=1}^{K}\|\mathbf{x}_{m,k}\|_1.
	\label{eq:multi_CBPDN}
\end{equation}
By defining
\begin{equation}
\begin{aligned}
	\mathbf{X}=&\begin{bmatrix}
	\mathbf{x}_{0,0}  & \dots & \mathbf{x}_{0,K} \\
	\vdots & \ddots & \vdots \\
	\mathbf{x}_{M,0} & \dots &  \mathbf{x}_{M,K}
	\end{bmatrix}\in\mathbb{C}^{MN\times K} \\
	\mathbf{S}=&\left[\mathbf{s}_0 \dots \mathbf{s}_K\right]\in\mathbb{C}^{N\times K} \\
	\mathbf{D}=&\left[\mathbf{D}_0 \dots \mathbf{D}_M\right]\in\mathbb{C}^{N\times MN},
\end{aligned}
\label{eq:concate}
\end{equation}
where $ \mathbf{D}_m $ is the matrix form of the convolutional operator that satisfies $ \mathbf{D}_m\mathbf{x}_m=\mathbf{d}_m\ast\mathbf{x}_m $, we can reformulate \eqref{eq:multi_CBPDN} in the form
\begin{equation}
	\argmin_{\mathbf{X}}\frac{1}{2}\|\mathbf{D}\mathbf{X}-\mathbf{S}\|_F^2+\lambda\|\mathbf{X}\|_{1,1}.
	\label{eq:compact_multiCBPDN}
\end{equation}

To solve the problem \eqref{eq:compact_multiCBPDN}, we utilize the \textit{Alternating Direction Method of Multipliers} (ADMM) \cite{boyd2011distributed,hong2019cospace,hong2019learnable,hong2019learning} in this paper. We refer the readers to \cite{garcia2018convolutional} for other algorithms to solve this problem. By introducing \textit{dual variable} $ \mathbf{U} $, penalty parameter $ \rho $ and auxiliary variable $ \mathbf{Y} $, the  corresponding scaled augmented Lagrangian function is defined as \cite{boyd2011distributed}:
\begin{equation}
\begin{aligned}
	L_\rho(\mathbf{X},\mathbf{Y},\mathbf{U})=&\frac{1}{2}\|\mathbf{D}\mathbf{X}-\mathbf{S}\|_F^2+\lambda\|\mathbf{Y}\|_{1,1}\\
	&+\frac{\rho}{2}\|\mathbf{X}-\mathbf{Y} + \mathbf{U}\|_F^2.
\end{aligned}
\label{eq:aug_lag_sparsecoding}
\end{equation}
Accordingly, the minimization of $ L_\rho $, with respect to each variable, can be solved by the following optimization subproblems:

1) $ \mathbf{X} $ subproblem for reconstructing sparse coefficients:
\begin{equation}
	\argmin_\mathbf{X}\frac{1}{2}\|\mathbf{D}\mathbf{X}-\mathbf{S}\|_F^2+\frac{\rho}{2}\|\mathbf{X}-\mathbf{Y}+\mathbf{U}\|_F^2.
\label{eq:convsparsecodeing_X_subproblem}
\end{equation}
The minimization of this quadratic function can be calculated by setting the associated derivative to zero, which leads to the following linear system:
\begin{equation}
	\left(\mathbf{D}^H\mathbf{D}+\rho\mathbf{I}\right)\mathbf{X}=\mathbf{D}^H\mathbf{S}+\rho(\mathbf{Y}-\mathbf{U}).
\label{eq:convsparsecodeing_X_subproblem_lin_sys}
\end{equation}
By applying Fast Fourier Transform (FFT) \cite{hong2020invariant}, this linear system can be efficiently solved in the frequency domain as:
\begin{equation}
	\left(\hat{\mathbf{D}}^H\hat{\mathbf{D}}+\rho\mathbf{I}\right)\hat{\mathbf{X}}=\hat{\mathbf{D}}^H\hat{\mathbf{S}}+\rho(\hat{\mathbf{Y}}-\hat{\mathbf{U}}),
\label{eq:convsparsecodeing_X_subproblem_lin_sys_DFT}
\end{equation}
where $ \hat{\mathbf{A}} $ denotes the DFT version of variable $ \mathbf{A} $. The solution of this linear system can be obtained by exploiting the Sherman-Morrison formula \cite{wohlberg2016efficient}.

2) $ \mathbf{Y} $ subproblem for calculating auxiliary variable:
\begin{equation}
	\argmin_\mathbf{Y}\lambda\|\mathbf{Y}\|_{1,1}+\frac{\rho}{2}\|\mathbf{X}-\mathbf{Y}+\mathbf{U}\|_F^2.
	\label{eq:convsparsecodeing_Y_subproblem}
\end{equation}
In real domain, this subproblem has a closed-form solution defined as:
\begin{equation}
	\mathbf{Y}:=\mathcal{S}_{\frac{\lambda}{\rho}}\left(\mathbf{X}+\mathbf{U}\right),
	\label{eq:convsparsecoding_softthres}
\end{equation}
where $ \mathcal{S}(\cdot) $ is the soft-thresholding function:
\begin{equation}
	\mathcal{S}_\gamma(\mathbf{A})=\mathrm{sign}(\mathbf{A})\odot\max(0,|\mathbf{A}|-\gamma),
\end{equation}
with the two element-wise operators $ \mathrm{sign}(\cdot) $ and $ |\cdot| $, and $ \odot $ denotes the element-wise multiplication.

However, the soft-thresholding function \eqref{eq:convsparsecoding_softthres} cannot be directly applied to the data in complex domain. Thus, in this paper, we propose the corresponding soft-thresholding function in complex domain, termed as \textit{complex soft-thresholding}. For $ \mathbf{A}\in\mathbb{C} $, we define $ |\mathbf{A}| $ in complex domain as $ \sqrt{\mathrm{real}(\mathbf{A})^2+\mathrm{imag}(\mathbf{A})^2} $, where $ \mathrm{real}(\cdot) $ and $ \mathrm{imag}(\cdot) $ extract real and imaginary parts of $ \mathbf{A} $, $\sqrt{\cdot}$ and $ (\cdot)^2 $ denote the square root and square, with element-wise manner, respectively. Furthermore, $\frac{\mathbf{A}}{|\mathbf{A}|}$ means element-wise division. Correspondingly, complex soft-thresholding function
is defined as:
\begin{equation}
	\mathcal{CS}_\gamma(\mathbf{A})=\frac{\mathbf{A}}{|\mathbf{A}|}\odot\max(0,|\mathbf{A}|-\gamma).
	\label{eq:cpx_soft_thresholding}
\end{equation}
The interpretation of complex soft-thresholding is demonstrated in \Fig \ref{fig:svt_csvt}. In complex domain, the operator preserves the direction of the complex vector and shrinks its associated amplitude by $ \gamma $. Specifically, within the range of $ \gamma $, the complex vector is projected to the original point and outside the circle, the phase of the complex vector is kept and its amplitude is shrinked.

\begin{figure}[!t]
    \centering
    \includegraphics[width=0.5\textwidth]{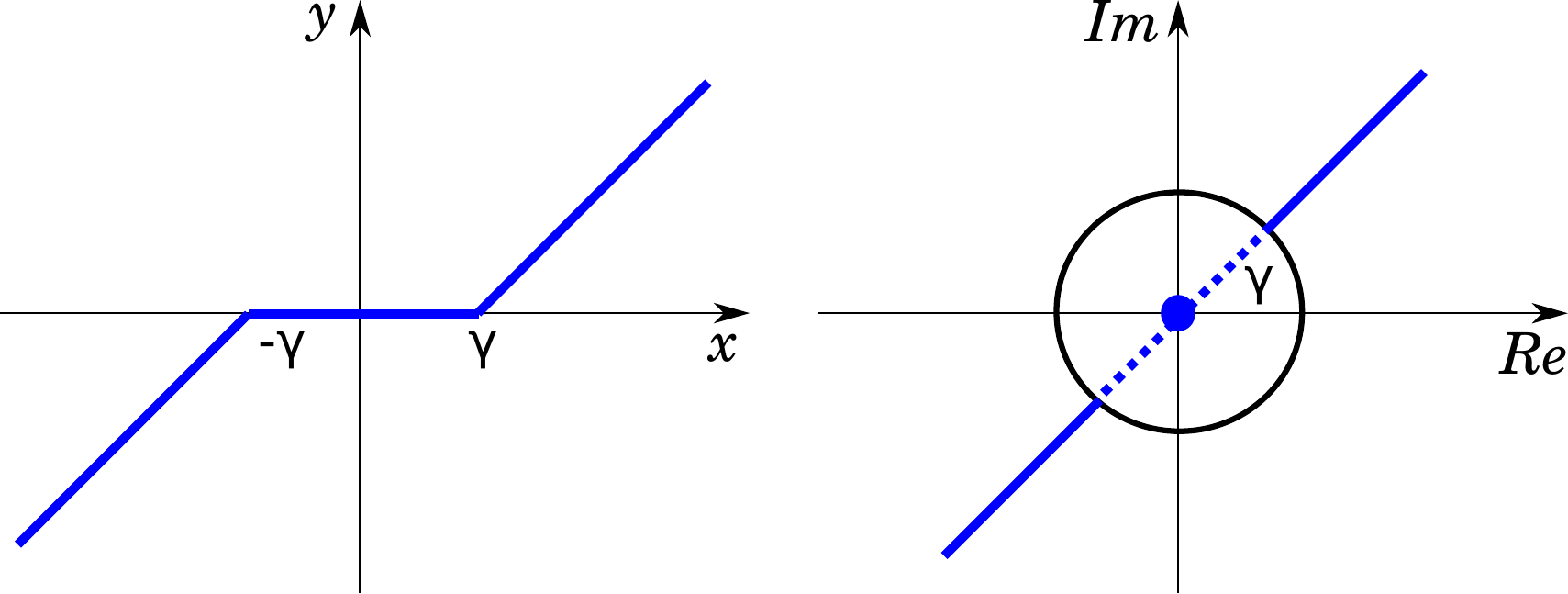}
    \caption{Soft-thresholding operator in real domain (Left). The proposed complex soft-thresholding operator for solving $ L_1 $ norm regularized optimization problem in complex domain. In complex domain, the operator preserves the direction of the complex vector and shrinks its associated amplitude by $ \gamma $. Specifically, within the range of $ \gamma $, the complex vector is projected to the original point and outside the circle, the phase of the complex vector is kept and its amplitude is shrinked.}
    \label{fig:svt_csvt}
\end{figure}

Actually, the closed-form solution of \eqref{eq:convsparsecodeing_Y_subproblem} in the complex domain is defined as\footnote{We will prove this conclusion in Appendix \ref{app:proof}.}:
\begin{equation}
\mathbf{Y}:=\mathcal{CS}_{\frac{\lambda}{\rho}}\left(\mathbf{X}+\mathbf{U}\right).
\label{eq:convsparsecoding_cpx_softthres}
\end{equation}

3) \textit{Multiplier update}: the dual variable $ \mathbf{U} $ can be updated as:
\begin{equation}
	\mathbf{U}:=\mathbf{U}+\mathbf{X}-\mathbf{Y}.
	\label{eq:sparse_coding_U_update}
\end{equation}
\subsection{Dictionary Update}
Fixing
the sparse coefficient maps\footnote{For convenience, we switch the indexing of $ \mathbf{x}_{m,k} $ as $ \mathbf{x}_{k,m} $} $ \{\mathbf{x}_{k,m}\} $ in (\ref{eq:cpx_conv_dl}), the dictionary update problem can be posed as:
\begin{equation}
	\argmin_{\mathbf{d}_m}\frac{1}{2}\sum_{k=1}^{K}\|\sum_{m=1}^{M}\mathbf{x}_{k,m}\ast\mathbf{d}_m-\mathbf{s}_k\|_2^2 \;\; {\rm s.t.} \;\; \|\mathbf{d}_m\|_2=1.
	\label{eq:conv_dict_update}
\end{equation}
In order to solve \eqref{eq:conv_dict_update} in the frequency domain, the $ L $ dimensional filter $ \mathbf{d}_m $ should be zero-padded to the common spatial dimensions of $ \mathbf{x}_{k,m} $ and $ \mathbf{s}_k $. $ \mathbf{P} $ denotes the projection operator that zeros the regions of the filters outside the desired support. We introduce the indicator function $ \iota_{C_{PN}} $
as: \begin{equation}
\begin{aligned}
	\iota_{C_{PN}}(\mathbf{d})=\begin{cases}
	0 \quad  & \mathrm{if} ~~\mathbf{d}\in  C_{PN}\\
	\infty  & \mathrm{if} ~~\mathbf{d} \notin  C_{PN}
	\end{cases},
\end{aligned}
\end{equation}
and the constraint set $ C_{PN} $ is denoted as:
\begin{equation}
	C_{PN}=\{\mathbf{d}\in\mathbb{C}^N: (\mathbf{I}-\mathbf{P}\mathbf{P}^T)\mathbf{d}=0,~\|\mathbf{d}\|_2=1\}.
\end{equation}
Similar to the formulation in sparse coding, the problem \eqref{eq:conv_dict_update} can be expressed as:
\begin{equation}
	\argmin_\mathbf{d}\frac{1}{2}\|\mathbf{X}\mathbf{d}-\mathbf{s}\|_2^2+\iota_{C_{PN}}(\mathbf{d}),
\label{eq:compact_multidict_update}
\end{equation}
where $ \mathbf{s} $ and $ \mathbf{d} $ are defined as:
\begin{equation}
\mathbf{s}=\begin{bmatrix}
\mathbf{s}_0\\
\vdots\\
\mathbf{s}_K
\end{bmatrix}\in\mathbb{C}^{NK} \quad
\mathbf{d} = \begin{bmatrix}
\mathbf{d}_0\\
\vdots\\
\mathbf{d}_M
\end{bmatrix}\in\mathbb{C}^{NM}.
\end{equation}
By introducing dual variable $ \mathbf{u} $, penalty parameter $ \sigma $ and auxiliary variable $ \mathbf{y} $, the corresponding scaled augmented Lagrangian function of \eqref{eq:compact_multidict_update} is defined as:
\begin{equation}
	L_\sigma(\mathbf{d},\mathbf{y},\mathbf{u})=\frac{1}{2}\|\mathbf{X}\mathbf{d}-\mathbf{s}\|_2^2+\iota_{C_{PN}}(\mathbf{y})+\frac{\sigma}{2}\|\mathbf{d}-\mathbf{y}+\mathbf{u}\|_2^2
	\label{eq:convsparse_dict_update_aug_lag}
\end{equation}
By applying ADMM, the minimization of $ L_\sigma $, with respect to each variable, can be solved by the following optimization subproblems:

1) $ \mathbf{d} $ subproblem for calculating complex convolutional filters:
\begin{equation}
	\argmin_\mathbf{d}\frac{1}{2}\|\mathbf{X}\mathbf{d}-\mathbf{s}\|_2^2+\frac{\sigma}{2}\|\mathbf{d}-\mathbf{y}+\mathbf{u}\|_2^2.
\end{equation}
Similar with the subproblem in \eqref{eq:convsparsecodeing_X_subproblem}, this minimization problem can be reformulated by solving the linear system as:
\begin{equation}
	\left(\mathbf{X}^H\mathbf{X}+\sigma\mathbf{I}\right)\mathbf{d}=\mathbf{X}^H\mathbf{s}+\sigma\left(\mathbf{y}-\mathbf{u}\right),
\end{equation}
which can also be transformed into the frequency domain as follows:
\begin{equation}
	\left(\hat{\mathbf{X}}^H\hat{\mathbf{X}}+\sigma\mathbf{I}\right)\hat{\mathbf{d}}=\hat{\mathbf{X}}^H\hat{\mathbf{s}}+\sigma\left(\hat{\mathbf{y}}-\hat{\mathbf{u}}\right).
	\label{eq:ComCSC_d_update}
\end{equation}
This linear system can be solved by Iterated Sherman-Morrison algorithm \cite{wohlberg2016efficient}.

2) $ \mathbf{y} $ subproblem for calculating the auxiliary variable:
\begin{equation}
	\argmin_\mathbf{y}\iota_{C_{PN}}(\mathbf{y})+\frac{\sigma}{2}\|\mathbf{d}-\mathbf{y}+\mathbf{u}\|_2^2,
\end{equation}
which is the proximal operator of the indicator function $ \iota_{C_{PN}} $ at the point $ \mathbf{d}+\mathbf{u} $. In particular, this proximal operator can be determined as:
\begin{equation}
y =	\mathrm{prox}_{\iota_{C_{PN}}}(\mathbf{d+u})=\frac{\mathbf{P}\mathbf{P}^T(\mathbf{d+u})}{\|\mathbf{P}\mathbf{P}^T(\mathbf{d+u})\|_2}.
	\label{eq:ComCSC_y_update}
\end{equation}

3) \textit{Multiplier update}: the dual variable $ \mathbf{u} $ can be updated by:
\begin{equation}
	\mathbf{u}:=\mathbf{u}+\mathbf{d}-\mathbf{y}.
	\label{eq:ComCSC_u_update}
\end{equation}

To this end, based on the previous optimization procedure, convolutional filters $ \mathbf{d}_m $ can be learned from the training interferograms. The algorithm for CCDL (\ref{eq:cpx_conv_dl}) is summarized in Algorithm \ref{ag:ADMM_solver_ComSCS}.


\begin{algorithm}[!t]
	\caption{CCDL solved by alternate minimization
	}
	\begin{algorithmic}[1]
		\Require Clean interferograms used for training: $ \{\mathbf{s}_k\}_{k=1}^K$
		\State Initialize $ \{\mathbf{d}_m \} $, $ L $, $ \lambda $,
		$\rho$, $\sigma$.
		\While { not convergent}
		\State \textbf{Sparse Coding}:
		\State Update $ \mathbf{X} $ by solving the subproblem in \eqref{eq:convsparsecodeing_X_subproblem_lin_sys_DFT}.
		\State Update $ \mathbf{Y} $ by complex soft-thresholding in \eqref{eq:convsparsecoding_cpx_softthres}.
		\State Update $ \mathbf{U} $ by \eqref{eq:sparse_coding_U_update}.
		\State \textbf{Dictionary Update}:
		\State Update $ \mathbf{d} $ by solving the subproblem in \eqref{eq:ComCSC_d_update}.
		\State Update $ \mathbf{y} $ by utilizing proximal operator of indicator function in \eqref{eq:ComCSC_y_update}.
		\State {Update $ \mathbf{u} $ by \eqref{eq:ComCSC_u_update}.}
		\EndWhile
		\Ensure $ \{\mathbf{d}_m \}_{m=1}^M $
	\end{algorithmic}
	\label{ag:ADMM_solver_ComSCS}
\end{algorithm}

\subsection{Computational Complexity}
\label{sec:ccdl_complexity}

The complexities of CCDL and the normal convolutional dictionary learning\cite{wohlberg2016efficient} are actually of the same order and listed below item by item.
\begin{itemize}
    \item \eqref{eq:convsparsecodeing_X_subproblem_lin_sys_DFT}: $\mathcal{O}(KMN) + \mathcal{O}(KMN\log(N))$.
    \item \eqref{eq:convsparsecoding_cpx_softthres}: $\mathcal{O}(KMN)$.
    \item \eqref{eq:sparse_coding_U_update}: $\mathcal{O}(KMN)$.
    \item \eqref{eq:ComCSC_d_update}: $\mathcal{O}(K^2MN) + \mathcal{O}(KMN\log(N))$.
    \item \eqref{eq:ComCSC_y_update}: $\mathcal{O}(KMN)$.
    \item \eqref{eq:ComCSC_u_update}: $\mathcal{O}(KMN)$.
\end{itemize}
Among the six steps of CCDL, \eqref{eq:ComCSC_d_update} has the dominant computational complexity. Therefore, the complexity of CCDL (Algorithm \ref{ag:ADMM_solver_ComSCS}) can be summarized as
\[\mathcal{O}\Big(T\big(K^2MN + KMN\log(N)\big)\Big ),\]where $T$ is the number of iterations that Algorithm \ref{ag:ADMM_solver_ComSCS} takes before convergence. Typically, $T = 200, K= 80, M = 96, N = 100 \times 100$, and Algorithm \ref{ag:ADMM_solver_ComSCS} takes around $ 3 $ hours to stop\footnote{The codes are implemented by MATLAB and the experiments are conducted on a PC with Intel® Core™ i7-8850H CPU @ 2.60GHz.}. However, we have to note that a common set of  filters  $\{\mathbf{d}_m \}_{m=1}^M$ can be used for all images in a dataset, not only for one image. Thus, we conduct CCDL only once and store the dictionary $\{\mathbf{d}_m \}$ before denoising since CCDL is expensive compared with ComCSC.

\section{Complex Convolutional Sparse Coding with Gradient Regularization}\label{sc:ccdl_GR}

Given the noisy interferogram $ \tilde{\mathbf{s}} $ and the learned filters $\{\mathbf{d}_m\}_{m=1}^M$, the sparse coefficient maps $ \{\tilde{\mathbf{x}}_m\} $ can be obtained by applying CBPDN and the restored interferogram $ \check{\mathbf{s}} $ can be described as the convolutional sparse representation, \textit{i.e.} $\check{\mathbf{s}}:=\sum_{m=1}^{M}\mathbf{d}_m\ast\mathbf{\tilde{\mathbf{x}}_m} $. In the following, we term this method as Complex Convolutional Sparse Coding  (\textit{ComCSC}).

\subsection{ComCSC-GR}

As introduced in \cite{wohlberg2016convolutional}, convolutional sparse representations can provide good restorations for high-pass components of images.
In order to also well reconstruct the low-pass components of images, gradient regularization on the sparse coefficient maps can be further exploited.
In this paper, aiming for adapting to both the high-pass and low-pass components of interferograms, we also extend ComCSC with the regularization of gradients, denoted as \textit{ComCSC-GR} and investigate its performance for interferometric phase restoration. Such model can be represented as\footnote{ In ComCSC-GR (\ref{eq:cbpdn_gr}), we should use notation $\tilde{\mathbf{s}}$ and $\{\tilde{\mathbf{x}}_m\}$ to represent noisy interferograms and the corresponding  sparse codes. In this section, we use ${\mathbf{s}}$ and $\{{\mathbf{x}}_m\}$ for simplicity. }:
\begin{equation}
\begin{aligned}
	\argmin_{\{\mathbf{x}_m \}}&\frac{1}{2}\|\sum_m\mathbf{d}_m\ast\mathbf{x}_m - \mathbf{s} \|_2^2 + \lambda\sum_m\|\mathbf{x}_m\|_1 +\\ &\frac{\mu}{2}\sum_m\|\sqrt{(\mathbf{g}_0\ast\mathbf{x}_m)^2+(\mathbf{g}_1\ast\mathbf{x}_m)^2}\|_2^2,
\end{aligned}
\label{eq:cbpdn_gr}
\end{equation}
where $ \mathbf{g}_0 $ and $ \mathbf{g}_1 $ are the filters which compute the gradients along image rows and columns, respectively. By introducing linear operators $ \mathbf{G}_0 $ and $ \mathbf{G}_1 $, \textit{i.e.} $ \mathbf{G}_l\mathbf{x}_m=\mathbf{g}_l\ast\mathbf{x}_m ~(l = 0,1) $, problem \eqref{eq:cbpdn_gr} can be rewritten as:
\begin{equation}
	\argmin_{\mathbf{x}}\frac{1}{2}\|\mathbf{D}\mathbf{x}-\mathbf{s}\|_2^2+\lambda\|\mathbf{x}\|_1+\frac{\mu}{2}(\|\boldsymbol{\Gamma}_0\mathbf{x}\|_2^2+\|\boldsymbol{\Gamma}_1\mathbf{x}\|_2^2),
\end{equation}
where
\begin{equation}
	\boldsymbol{\Gamma}_l=\begin{bmatrix}
	\mathbf{G}_l & 0 & \cdots \\
	0 & \mathbf{G}_l & \cdots \\
	\vdots & \vdots & \ddots
	\end{bmatrix}\in\mathbb{C}^{MN\times MN} \quad
	\mathbf{x} = \begin{bmatrix}
	\mathbf{x}_0\\
	\vdots\\
	\mathbf{x}_M
	\end{bmatrix}\in\mathbb{C}^{NM}.
\end{equation}
This problem can also be solved by introducing dual variable $ \mathbf{u} $, auxiliary variable $ \mathbf{y} $ and the associated penalty parameter $ \rho $ in the ADMM optimization procedure. The corresponding
scaled
augmented Lagrangian function is defined as:
\begin{equation}
\begin{aligned}
	L_\rho(\mathbf{x},\mathbf{y},\mathbf{u})=&\frac{1}{2}\|\mathbf{D}\mathbf{x}-\mathbf{s}\|_2^2+\lambda\|\mathbf{y}\|_1+\frac{\mu}{2}(\|\boldsymbol{\Gamma}_0\mathbf{x}\|_2^2+\|\boldsymbol{\Gamma}_1\mathbf{x}\|_2^2)\\
	&+\frac{\rho}{2}\|\mathbf{x}-\mathbf{y}+ \mathbf{u}\|_2^2.
\end{aligned}
\end{equation}
The subproblems of $ L_\rho $ with respect to each variable can be written as follows:

1) $ \mathbf{x} $ subproblem for reconstructing the sparse coefficients:
\begin{equation}
\begin{aligned}
\argmin_{\mathbf{x}}&\frac{1}{2}\|\mathbf{D}\mathbf{x}-\mathbf{s}\|_2^2+\frac{\mu}{2}(\|\boldsymbol{\Gamma}_0\mathbf{x}\|_2^2+\|\boldsymbol{\Gamma}_1\mathbf{x}\|_2^2)\\
&+\frac{\rho}{2}\|\mathbf{x}-\mathbf{y}+\mathbf{u}\|_2^2.
\end{aligned}
\end{equation}
The solution of this subproblem can be obtained by solving the following linear system in the frequency domain:
\begin{equation}
\begin{aligned}
    \left(\hat{\mathbf{D}}^H\hat{\mathbf{D}}+\mu\hat{\boldsymbol{\Gamma}}_0^H\hat{\boldsymbol{\Gamma}}_0+\mu\hat{\boldsymbol{\Gamma}}_1^H\hat{\boldsymbol{\Gamma}}_0+\rho\mathbf{I}\right)\hat{\mathbf{x}}\\
    =\hat{\mathbf{D}}^H\hat{\mathbf{s}}+\rho(\hat{\mathbf{y}}-\hat{\mathbf{u}}).
\end{aligned}
	\label{eq:ComCSC-GR_x_update}
\end{equation}

2) $ \mathbf{y} $ subproblem for calculating the auxiliary variable:
\begin{equation}
	\argmin_\mathbf{y}\lambda\|\mathbf{y}\|_1+\frac{\rho}{2}\|\mathbf{x}-\mathbf{y}+\mathbf{u}\|_2^2.
	\label{eq:ComCSC-GR_y_update}
\end{equation}
Identically to \eqref{eq:convsparsecodeing_Y_subproblem}, this subproblem can be solved by the soft-thresholding in complex domain.

3) \textit{Multiplier update}: the dual variable $ \mathbf{u} $ can be updated by:
\begin{equation}
	\mathbf{u}:=\mathbf{u}+\mathbf{x}-\mathbf{y}.
	\label{eq:ComCSC-GR_u_update}
\end{equation}

The corresponding pseudocode of ComCSC and  ComCSC-GR are summarized in Algorithm \ref{ag:ADMM_solver_ComCSC} and \ref{ag:ADMM_solver_ComSCS-GR} respectively.

{
\begin{algorithm}[!t]
	\caption{ComCSC solved by ADMM}
	\begin{algorithmic}[1]
		\Require $ \mathbf{s}$, $ \{\mathbf{d}_m \} $
		\State Initialize $ \lambda $, $ \mu $, $ \rho $, let $\mathbf{g}_0=0$, $\mathbf{g}_1=0$.
		\While {not convergent}
		\State Update $ \mathbf{x} $ by
		\eqref{eq:ComCSC-GR_x_update}.
		\State Update $ \mathbf{y} $ by
		\eqref{eq:ComCSC-GR_y_update}.
		\State Update $ \mathbf{u} $ by \eqref{eq:ComCSC-GR_u_update}.
		\EndWhile
		\Ensure
$\check{\mathbf{s}}:=\sum_{m=1}^{M}\mathbf{d}_m\ast\mathbf{{\mathbf{x}}_m} $
	\end{algorithmic}
	\label{ag:ADMM_solver_ComCSC}
\end{algorithm}

}

\begin{algorithm}[!t]
	\caption{ComCSC-GR solved by ADMM}
	\begin{algorithmic}[1]
		\Require $ \mathbf{s}$, $ \{\mathbf{d}_m \} $, $\mathbf{g}_0$, $\mathbf{g}_1$
		\State Initialize $ \lambda $, $ \mu $, $ \rho $.
		\While {not convergent}
		\State Update $ \mathbf{x} $ by solving the subproblem in \eqref{eq:ComCSC-GR_x_update}.
		\State Update $ \mathbf{y} $ by complex soft-thresholding in \eqref{eq:ComCSC-GR_y_update}.
		\State Update $ \mathbf{u} $ by \eqref{eq:ComCSC-GR_u_update}.
		\EndWhile
		\Ensure
$\check{\mathbf{s}}:=\sum_{m=1}^{M}\mathbf{d}_m\ast\mathbf{{\mathbf{x}}_m} $
	\end{algorithmic}
	\label{ag:ADMM_solver_ComSCS-GR}
\end{algorithm}

In a similar vein, given the noisy interferogram $ \tilde{\mathbf{s}} $ and the filters learned by CCDL, the sparse coefficient maps $ \{\tilde{\mathbf{x}}_m\} $ can be obtained by applying ComCSC-GR and the restored interferogram $ \check{\mathbf{s}} $ can be described as the convolutional sparse representation, \textit{i.e.} $\check{\mathbf{s}}:=\sum_{m=1}^{M}\mathbf{d}_m\ast\mathbf{\tilde{\mathbf{x}}_m} $. Comparing to ComCSC, its gradient-regularized version can not only maintain the performance of the reconstruction results on the high-pass components, but also the low-pass components can be well restored. Such advantages can be beneficial for the InSAR phase restoration, since the observed scene is usually characterized by both low- and high-frequency phase variations.

\subsection{Computational Complexity}
\label{sec:csc_complexity}
Similar with the argument in Section \ref{sec:ccdl_complexity}, the complexities of ComCSC-GR and CSC-GR\cite{wohlberg2016convolutional} are in the same order. We list them item by item here\footnote{The bounds given here are theoretical upper bounds for the worse cases. In practice, the computational cost is also related with the filter size $L$: smaller $L$ leads to a smaller cost than the upper bound. This phenomenon is due to the implementation of zero-padding and FFT. We will discuss this point in the text following \Fig \ref{fg:effi_sty_filter_num_sz}. }.
\begin{itemize}
    \item \eqref{eq:ComCSC-GR_x_update}: $\mathcal{O}(MN) + \mathcal{O}(MN\log(N))$,
    \item \eqref{eq:ComCSC-GR_y_update}: $\mathcal{O}(MN)$,
    \item \eqref{eq:ComCSC-GR_u_update}: $\mathcal{O}(MN)$,
\end{itemize}
where \eqref{eq:ComCSC-GR_x_update} has the dominant computational complexity. Therefore, the complexity of ComCSC-GR (Algorithm \ref{ag:ADMM_solver_ComSCS-GR}) can be summarized as
\[\mathcal{O}\Big(T MN\log(N)\Big ),\]where $T$ is the number of iterations that Algorithm \ref{ag:ADMM_solver_ComSCS-GR} takes before convergence. Typically, $T = 150, M = 96, N = 256 \times 256$, and Algorithm \ref{ag:ADMM_solver_ComSCS-GR} takes $ 140 $ seconds to stop.

\begin{figure}[!t]
	\centering
	\includegraphics[width=0.5\textwidth]{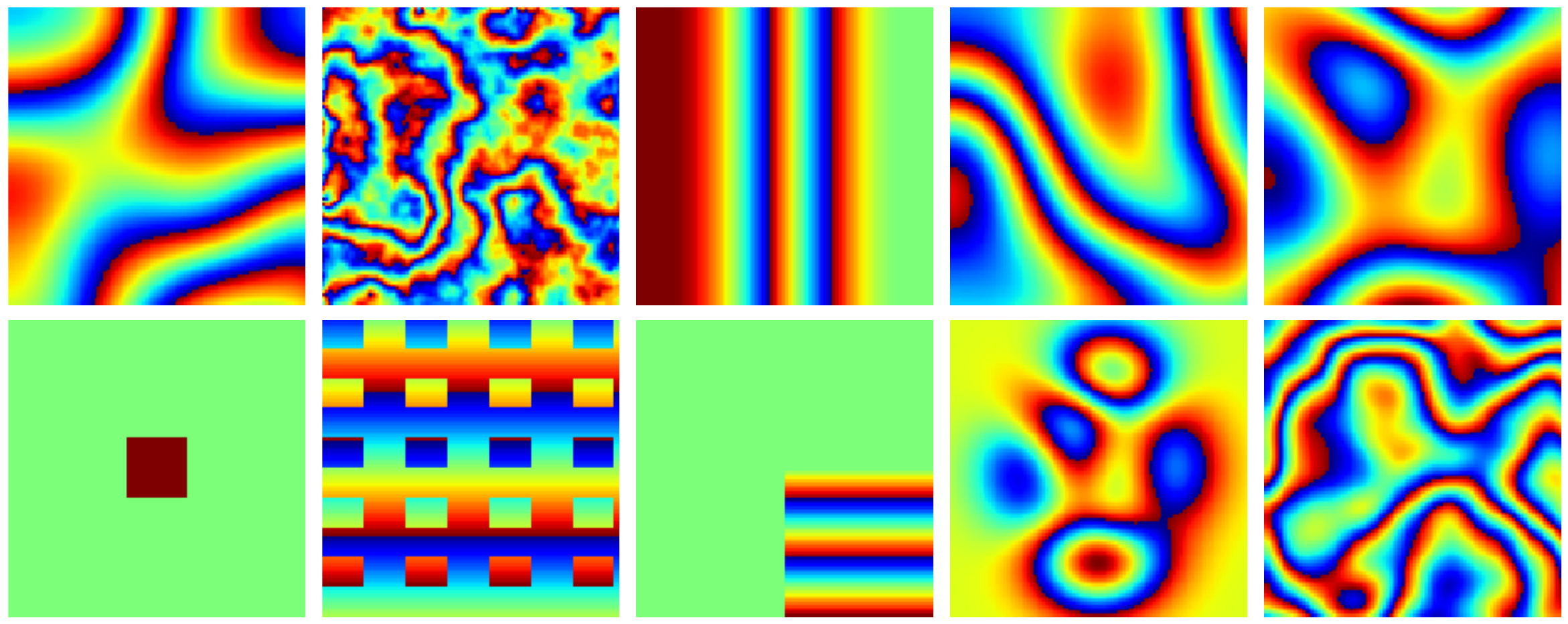}
	\caption{Interferogram examples utilized for complex convolutional dictionary learning.}
	\label{fg:gt_simu_cpxs}
\end{figure}

\begin{figure}[!t]
	\centering
	\includegraphics[width=0.5\textwidth]{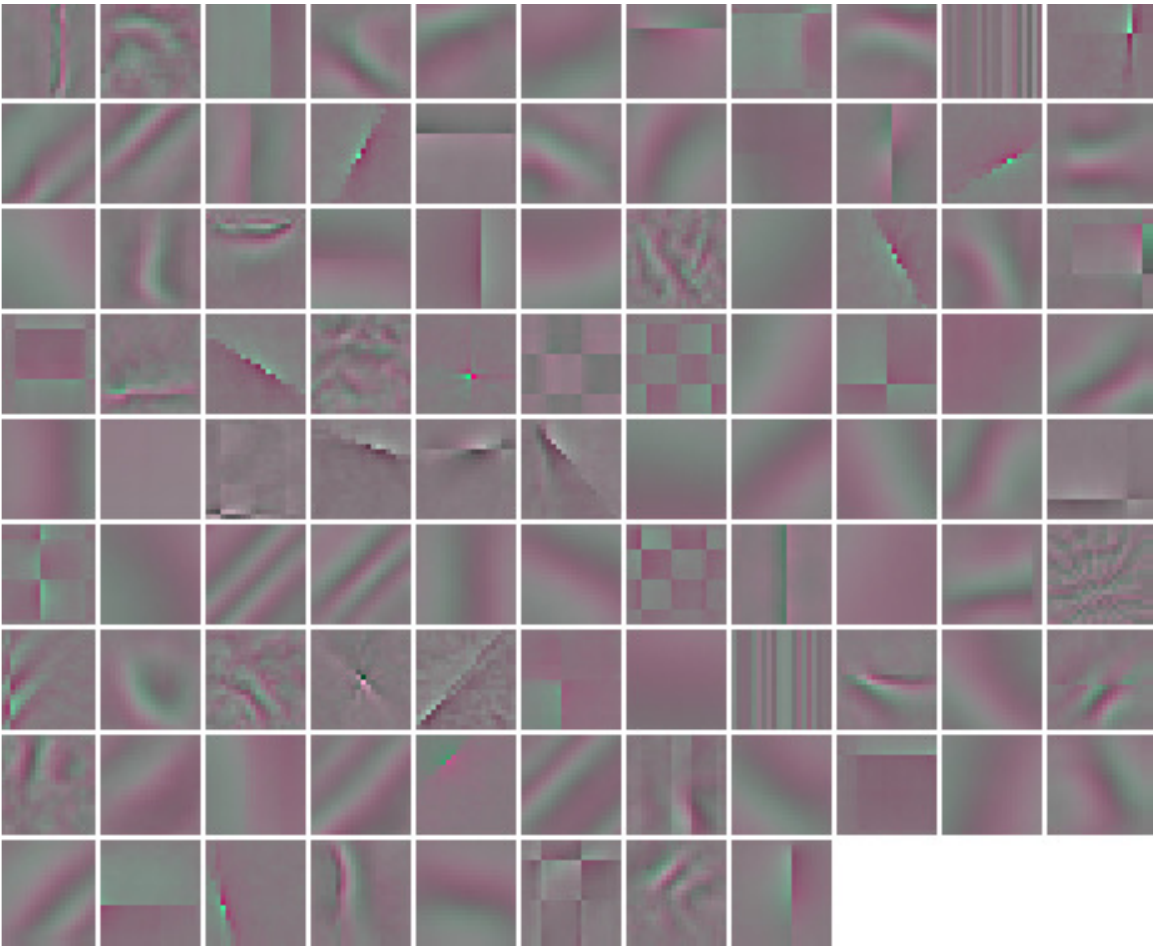}
	\caption{$ 96 $ convolutional filters learned on the ground truth interferograms. For the visualization of complex data, the real and imaginary parts are mapped into red and green colors, respectively.}
	\label{fg:lrn_conv_filters_96}
\end{figure}

\begin{figure}[!t]
	\centering
	\includegraphics[width=0.23\textwidth]{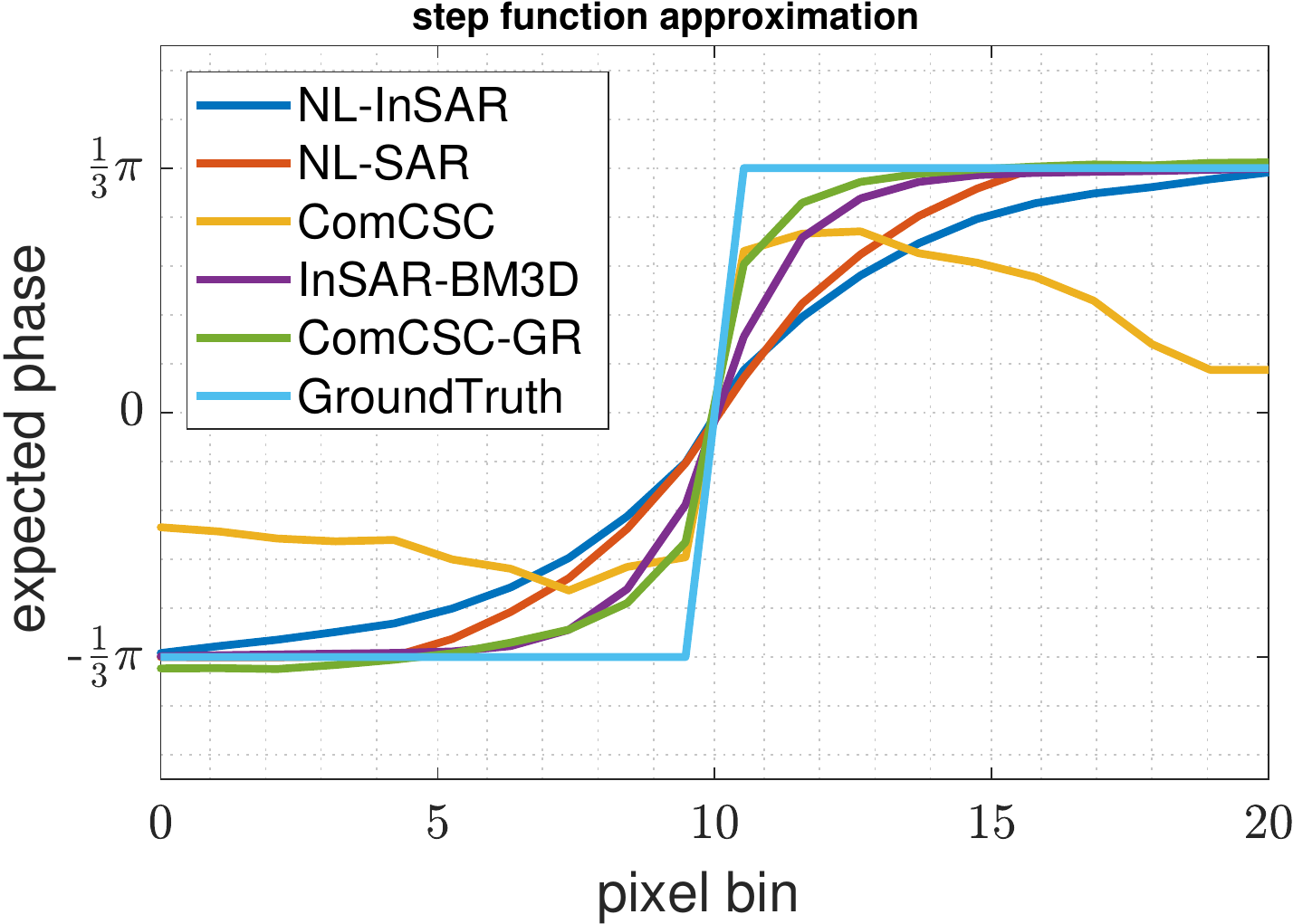}~
	\includegraphics[width=0.23\textwidth]{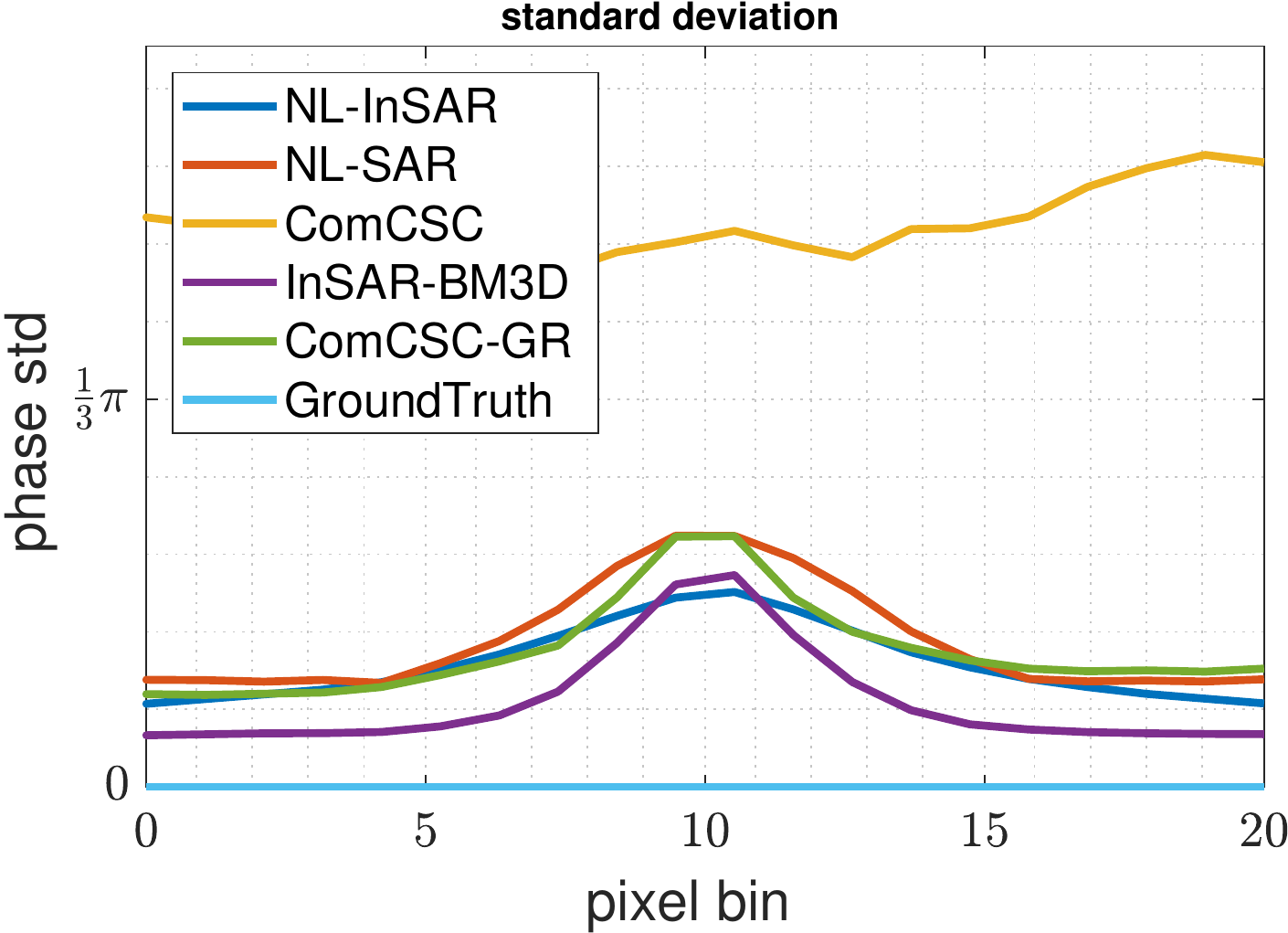}
	\caption{$ 10,000 $ Monte-Carlo simulations for evaluating the compared methods on the expected values and standard deviations of step function approximation. The amplitude is constant and the coherence value is set as $ 0.3 $.}
	\label{fg:step_func_appro}
\end{figure}

\section{Experimental Results}\label{sc:experiments}

\subsection{Simulations}
The covariance matrix of two correlated complex normal distributed scatterers of two Single Look Complex images (SLC) is defined as:
\begin{equation}
	\mathbf{C}=\begin{bmatrix}
	a^2 & a^2\gamma e^{\mathrm{j}\phi} \\
	a^2\gamma e^{-\mathrm{j}\phi} & a^2
	\end{bmatrix},
	\label{eq:37}
\end{equation}
where $ a $ is the amplitude, $ \phi $ denotes the interferometric phase and $ \gamma $ is the coherence. Given two independent complex normal distributed scatterers $ r_1 $ and $ r_2 $ of zero mean and unit variance, the synthetic samples $ u_1 $ and $ u_2 $ with the desired correlation can be obtained by
\begin{equation}
	\begin{bmatrix}
	u_1\\
	u_2
	\end{bmatrix}=\mathbf{L}\begin{bmatrix}
	r_1\\
	r_2
	\end{bmatrix}=a\begin{bmatrix}
	1 & 0 \\
	\gamma e^{-\mathrm{j}\phi} & \sqrt{1-\gamma^2}
	\end{bmatrix}\begin{bmatrix}
	r_1\\
	r_2
	\end{bmatrix},
	\label{eq:38}
\end{equation}
where $ \mathbf{L} $ represents the Cholesky decomposition of $ \mathbf{C} $. Then, the interferogram can be available by $ s=u_1\times \mathrm{conj}(u_2)$, where $ \mathrm{conj}(\cdot)$ is the conjugation operator.

In order to learn the convolutional filters $ \mathbf{d}_m $, we simulate a benchmark dataset of $ 80 $ interferograms $ \{\mathbf{s}_k\}_{k=1}^{K=80}$ with different patterns. Some examples are illustrated in \Fig \ref{fg:gt_simu_cpxs}. The spatial size of filters is set as $ 20\times20 $ pixels and the number of filters is $ 96 $. The parameter $ \lambda $ can be set to a small value, since the dataset utilized here is noiseless. Thus, we selected it as  $ 0.2 $. Based on CCDL, the learned filters are shown in \Fig \ref{fg:lrn_conv_filters_96}. For the visualization of complex data, the real and imaginary parts are mapped into red and green colors, respectively. It can be obviously observed that reliable interferometric phase components, such as curves, lines, rectangles, and smooth planes, can be learned based on the proposed method. The following experiments are based on such learned filters.

For evaluating the performance of the proposed method, we first compare it with other state-of-the-art methods by restoring a step function.
This experiment can give us an intuition about the various filters' capabilities of resolution and detailed structure preservation.
$ 10,000 $ Monte-Carlo simulations of the phase step function from $ -\frac{\pi}{3} $ to $ \frac{\pi}{3} $, with the coherence of $ 0.3 $ and the constant amplitude, are made for the experiment. The parameters of the referenced methods are introduced as follows:
\begin{itemize}
	\item The window size for Boxcar is $ 5\times5 $.
	\item The patch size and $ \alpha $ in Goldstein filter \cite{goldstein1998radar} is $ 16 $ and $ 0.9 $, respectively.
	\item The parameters are automatically chosen as stated in NL-SAR \cite{deledalle2015nl}.
	\item The search window and patch sizes in NL-InSAR \cite{deledalle2011nl} are $ 21 $ and $ 5 $, respectively.
	\item The parameters of InSAR-BM3D \cite{sica2018insar} are set the same as the original paper.
\end{itemize}

As shown in \Fig \ref{fg:step_func_appro} (Left), under the coherence of $ 0.3 $, the proposed ComCSC-GR can achieve the best approximation of the step function. ComCSC-GR outperforms the state-of-the-art InSAR filters, \textit{i.e.} InSAR-BM3D, in terms of the detail preserving of fringes. Although phase jump can also be well modeled by the ComCSC, the homogeneous areas cannot be restored by it. The reason is that high-pass components can be well reconstructed by convolutional sparse coding, with the sacrifice of homogeneous components. In comparison, with the gradient regularization of sparse coefficient maps, low-pass components can be very well preserved by ComCSC-GR. Therefore, both the homogeneous and phase jump areas can be well modeled by ComCSC-GR. Besides, both NL-InSAR and NL-SAR demonstrate their weakness on the edge preservation due to their intentionally oversmooth behaviors and no guidance filtering based on amplitude change in such areas. As illustrated in its right, without the low-pass regularization, the ComCSC performs the worst on the variances of the restoration. For the homogeneous area, the InSAR-BM3D achieves the best performance in terms of the stability for the restoration. For the phase jump area, owing to the proficient detail preservation, ComCSC-GR narrows the variances better than the NL-SAR.

\begin{figure}[!t]
	\centering
	\includegraphics[width=0.5\textwidth]{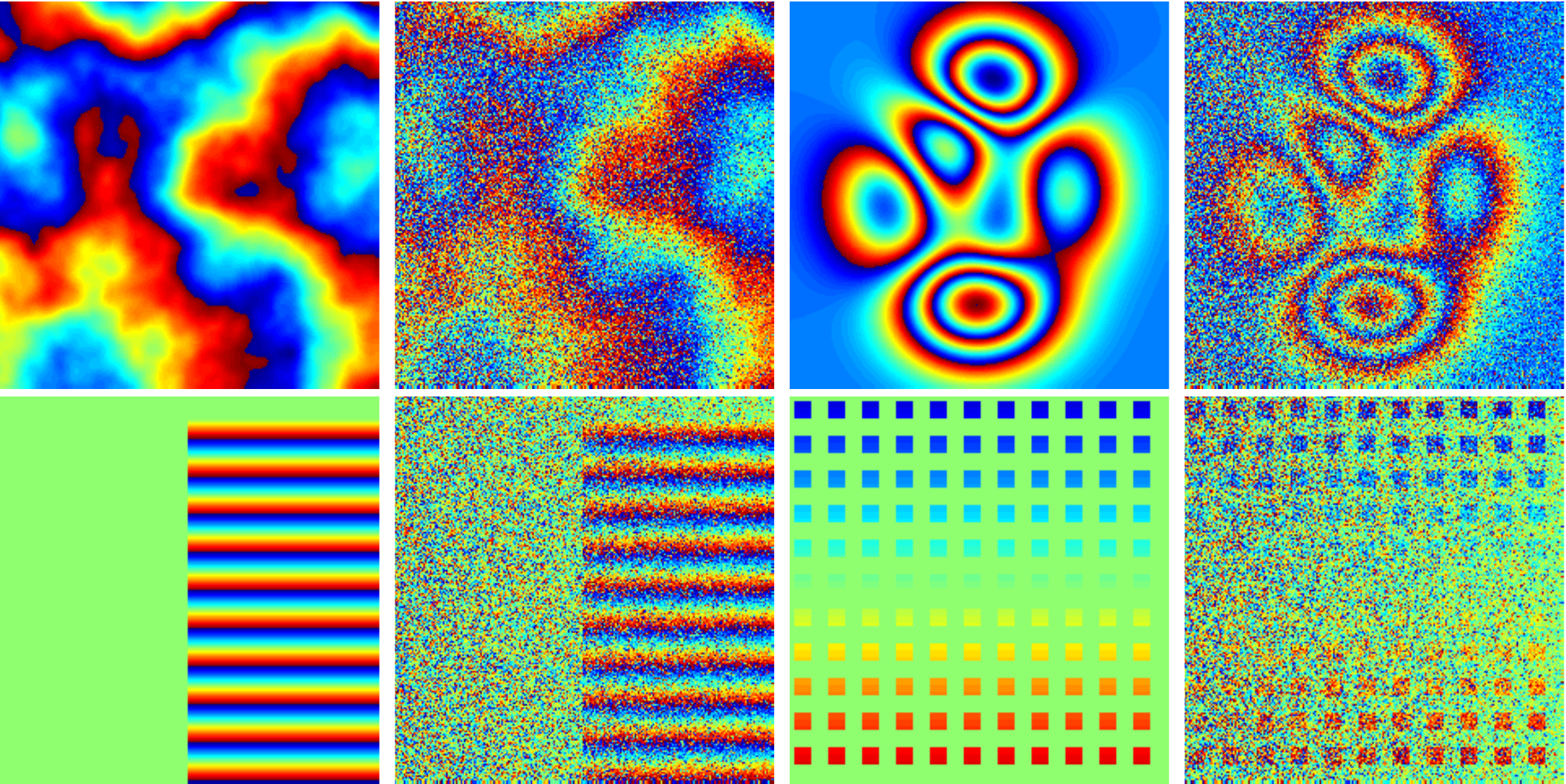}
	\caption{Simulated ground truth interfergrams, \textit{i.e.} mountain, peaks, shear plane and squares, and their noisy version. The coherence grows linearly from $ 0.3 $ (leftmost) to $ 0.9 $ (rightmost).}
	\label{fg:simu_gt_cpx}
\end{figure}

\begin{figure*}[!t]
	\centering
	\includegraphics[width=0.98\textwidth]{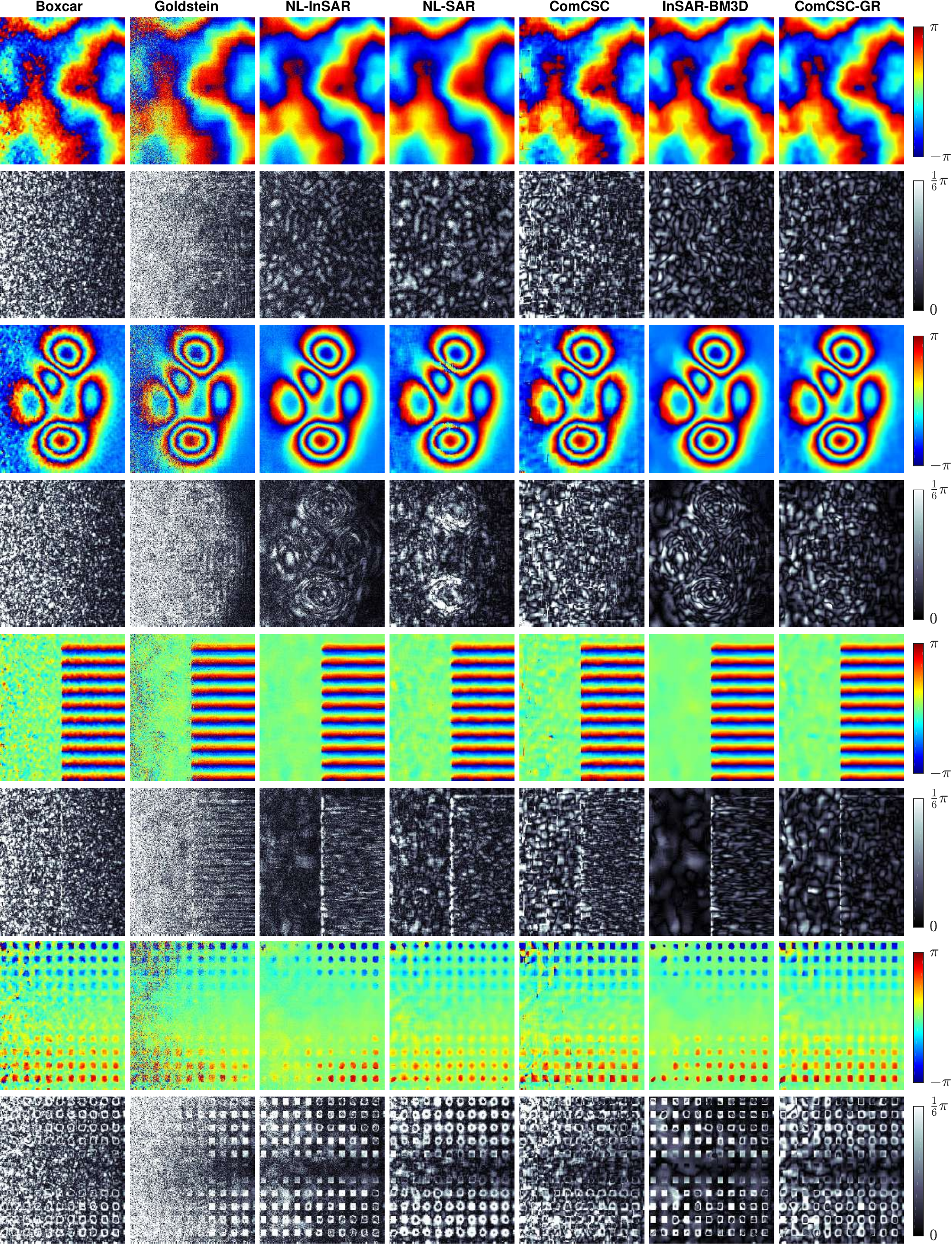}
	\caption{Filtered interferograms based on several comparing algorithms and their residual phases referenced with the ground truth images. }
	\label{fg:simu_4cpx_comp_results}
\end{figure*}

\begin{figure*}[!t]
	\centering
	\includegraphics[width=\textwidth]{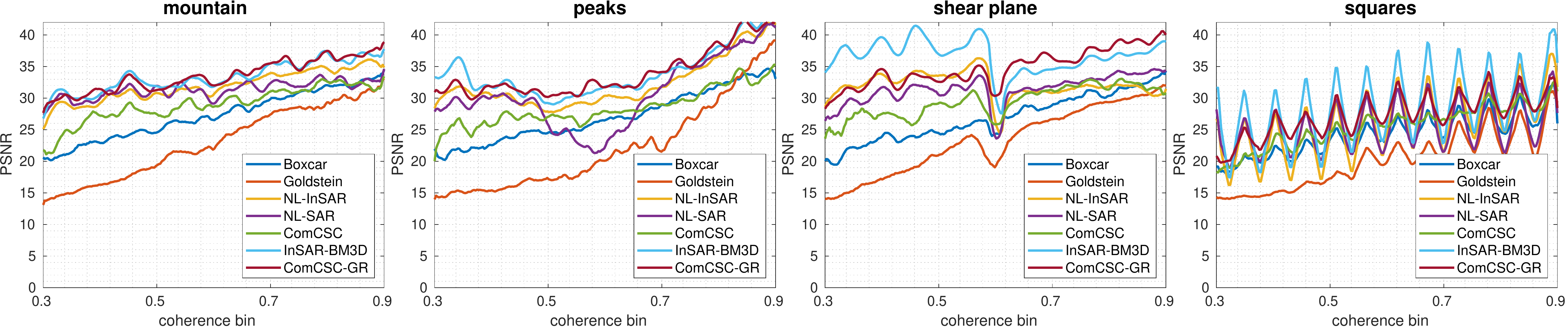}
	\caption{PSNR values calculated on the restored images in a sliding window manner with respect to the coherence values.}
	\label{fg:psnr_4cpx_wrt_coh}
\end{figure*}

For a more exhaustive analysis, we generate four different images with $ 256\times256 $ pixel size with typical interferometric patterns: 1) \textit{moutain} mimics the interferometric phases in mountainous terrains, 2) \textit{peaks} simulates a complex scenery with varied geometry, 3) \textit{shear plane} includes constant phase and rapid phase variation, 4) \textit{squares} replicates phase jump occurred in urban areas. Based on \eqref{eq:37} and \eqref{eq:38}, the noisy interferograms with spatially varied coherence map are generated. The coherence grows linearly from $ 0.3 $ (leftmost) to $ 0.9 $ (rightmost). The ground truth and noisy interferograms are depicted in \Fig \ref{fg:simu_gt_cpx}.

For a visual comparison, all the reconstruction results and the associated residual phases are shown in \Fig \ref{fg:simu_4cpx_comp_results}. Besides, we also make a quantitative evaluation based on the metric of \textit{peak signal-to-noise ratio} (PSNR) \cite{hongxing2015interferometric} defined as:
\begin{equation}
\mathrm{PSNR}:=10\log_{10}\frac{4N\pi^2}{\|\mathrm{angle}(\mathrm{conj}(\mathbf{s})\odot\check{\mathbf{s}})\|_F^2}[\mathrm{dB}].
\end{equation}
All the PSNR values of the overall images are illustrated in Table \ref{tb:simu_PSNR}. Moreover, as shown in \Fig \ref{fg:psnr_4cpx_wrt_coh}, we plot the PSNR values calculated on the restored images in a sliding window manner with respect to the coherence values. As observed in \Fig \ref{fg:simu_4cpx_comp_results}, for high coherence areas, all the methods can give reliable restoration results, since the residual phase maps are dark on such areas. However, the performances of all the methods are varied in the low coherence areas. For example, Boxcar and Goldstein cannot recover the interferometric patterns in such areas. All the other methods demonstrate better robustness in terms of coherence variation. Consistent with the previous experiment, oversmooth phenomenon can be found in the phase jump areas of nonlocal-based approaches, \textit{i.e.} NL-InSAR, NL-SAR and InSAR-BM3D. For example, as shown by the residual maps of peaks, contours of edges are displayed in the results of NL-InSAR, NL-SAR and InSAR-BM3D. Similarly, the phase change line can be also clearly seen in the reconstructed results of shear plane, as indicated in the corresponding residual maps.
In contrast, those edge areas can be better preserved by ComCSC-GR, also indicated by its higher PSNR values in comparison to all other methods, as shown in \Fig \ref{fg:psnr_4cpx_wrt_coh}.
From the numerical analysis in Table \ref{tb:simu_PSNR}, among all the methods, ComCSC-GR can achieve the best performance not only for the homogeneous pattern e.g. mountain, but also the heterogeneous pattern e.g. squares. It is worth noting that, for the patterns except squares, the performances of ComCSC-GR and InSAR-BM3D are comparable. However, for the sharply changing area (squares), ComCSC-GR can surpass the other nonlocal-based methods with a large margin (around 2dB), demonstrating the superiority of the proposed ComCSC-GR in modeling the sharp changes of fringes based on its learned convolutional kernels. To further investigate the performances of the comparing methods, the filtered interferometric results are transformed into the unwrapped phases with the same unwrapping method. As demonstrated in \Fig \ref{fg:simu_unwrapping}, the staircase effect can be evidently observed in both NL-InSAR and NL-SAR methods, especially in the examples of peaks and shear plane. Compared with InSAR-BM3D, continuous variation of real phases can be smoothly reconstructed by ComCSC-GR. As illustrated in the linear increasing part of shear plane and the bell-shape area of peaks, the continuous variation of the real phases can be more smoothly reconstructed by ComCSC-GR than InSAR-BM3D. Moreover, for the example of squares, with the powerful capability of detail preservation, ComCSC-GR can reconstruct more correct squares compared with the other methods.

\begin{table}[!t]
	\centering
	\caption{Numerical analysis of the compared methods on the simulated dataset}
	\label{tb:simu_PSNR}
	\begin{tabular}{*{4}{c|} c}
		\hline
		\hline
		& \multicolumn{4}{|c}{PSNR[dB]} \\
		\cline{2-5}
		& mountain & peaks & shear plane & squares \\
		\hline
		Boxcar & $ 25.61 $ & $ 25.35 $ & $ 25.21 $ & $ 22.57 $ \\
		\hline
		Goldstein & $ 19.79 $ & $ 18.69 $ & $ 20.29 $ & $ 18.04 $ \\
		\hline
		NL-InSAR & $ 31.26 $ & $ 31.00 $ & $ 30.58 $ & $ 21.79 $ \\
		\hline
		NL-SAR & $ 31.08 $ & $ 27.62 $ & $ 29.82 $ & $ 22.80 $ \\
		\hline
		ComCSC & $ 27.87 $ & $ 27.59 $ & $ 28.50 $ & $ 23.94 $ \\
		\hline
		InSAR-BM3D & $ 32.77 $ & $ 32.68 $ & $ 33.69 $ & $ 23.16 $ \\
		\hline
		ComCSC-GR & $ \mathbf{32.79} $ & $ \mathbf{32.71} $ & $ \mathbf{33.70} $ & $ \mathbf{25.36} $ \\
		\hline
		\hline
	\end{tabular}
	
\end{table}

\begin{figure}[!t]
 	\centering
 	\includegraphics[width=0.5\textwidth]{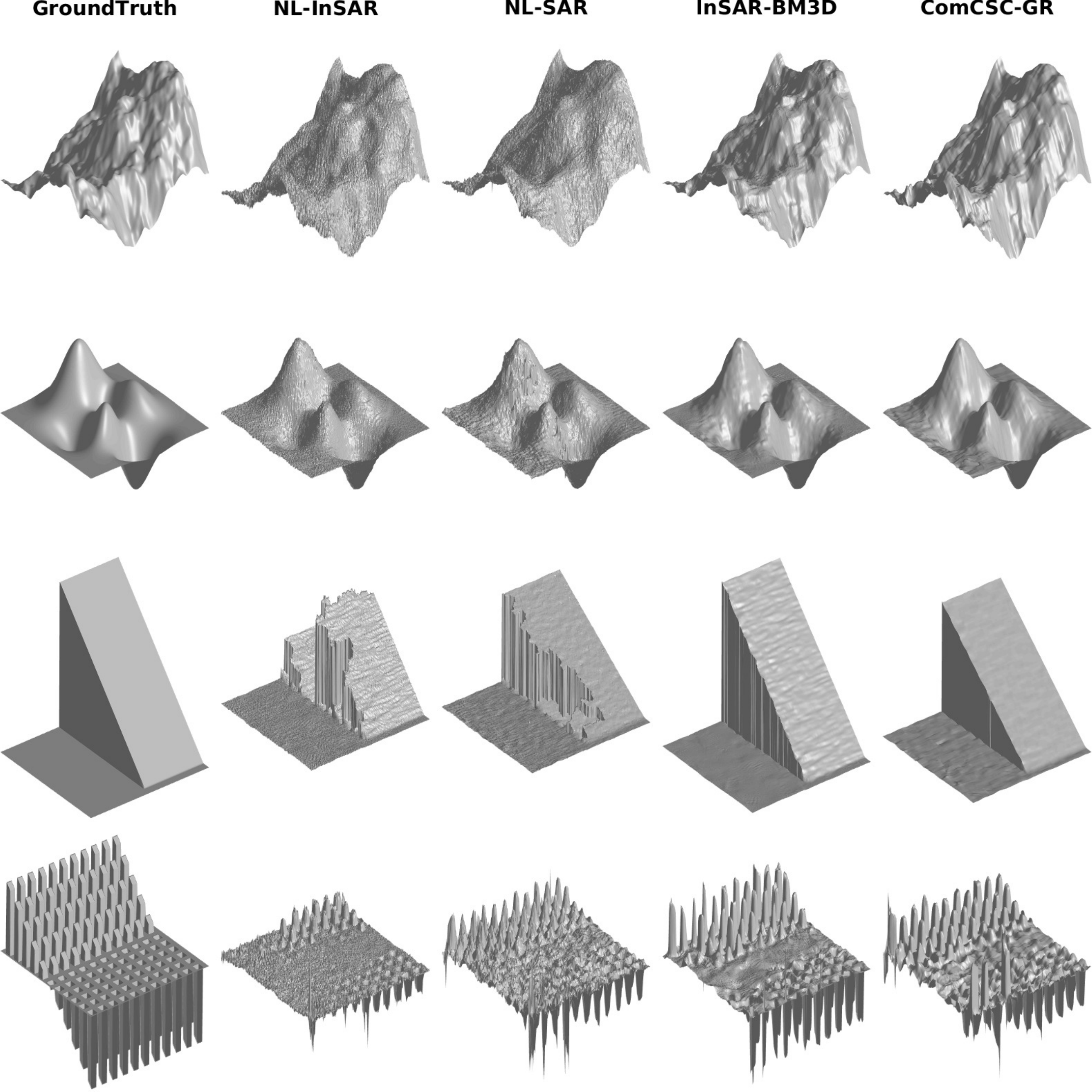}
 	\caption{Unwrapped phases of the filtered interferograms by the same phase unwrapping algorithm. }
 	\label{fg:simu_unwrapping}
\end{figure}

\begin{figure}[!t]
	\centering
	\includegraphics[width=0.24\textwidth]{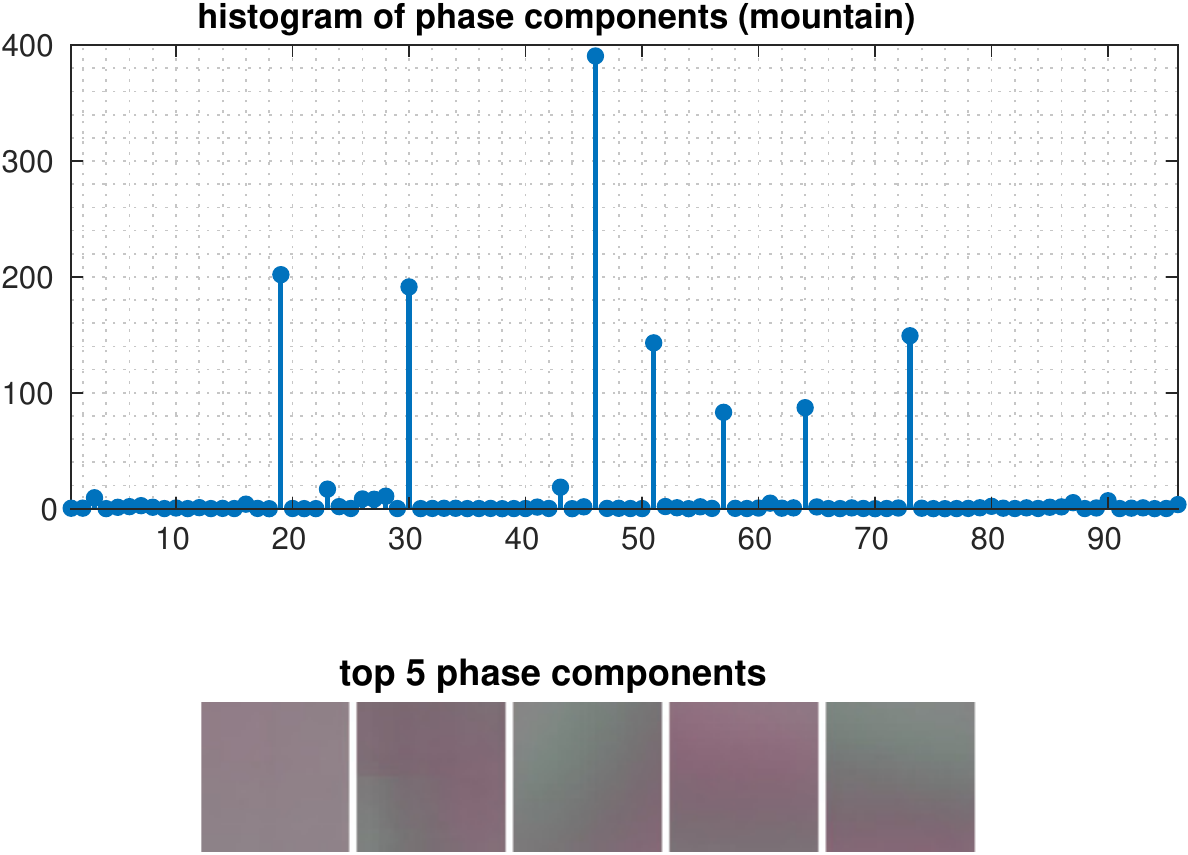}~
	\includegraphics[width=0.24\textwidth]{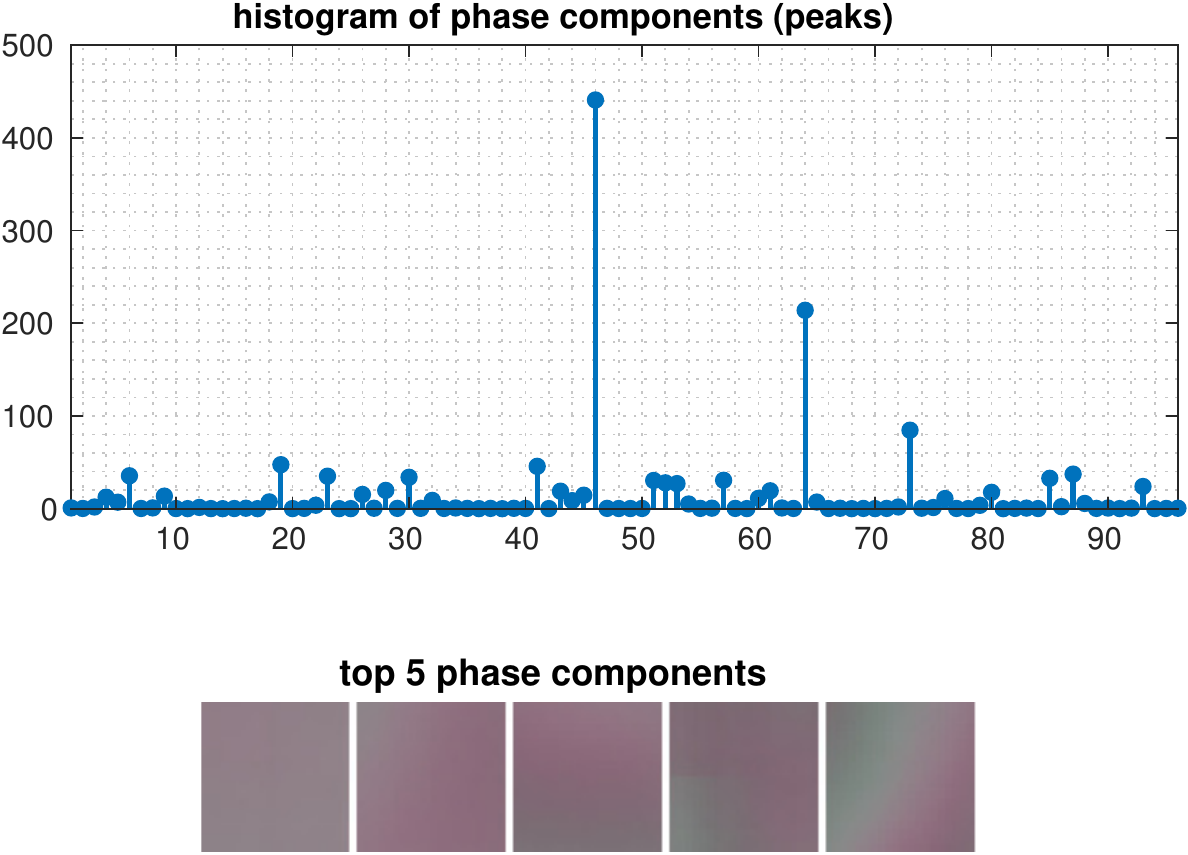}
	\hfill
	\includegraphics[width=0.24\textwidth]{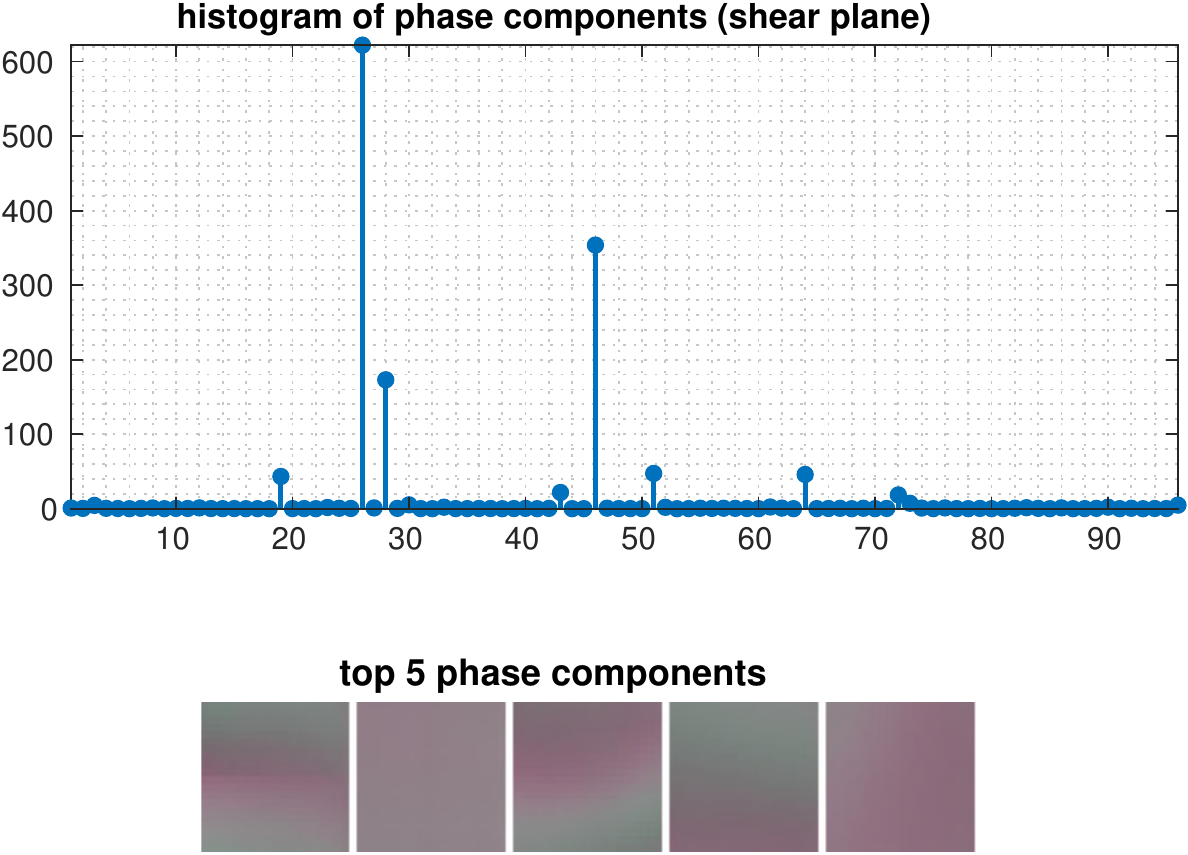}~
	\includegraphics[width=0.24\textwidth]{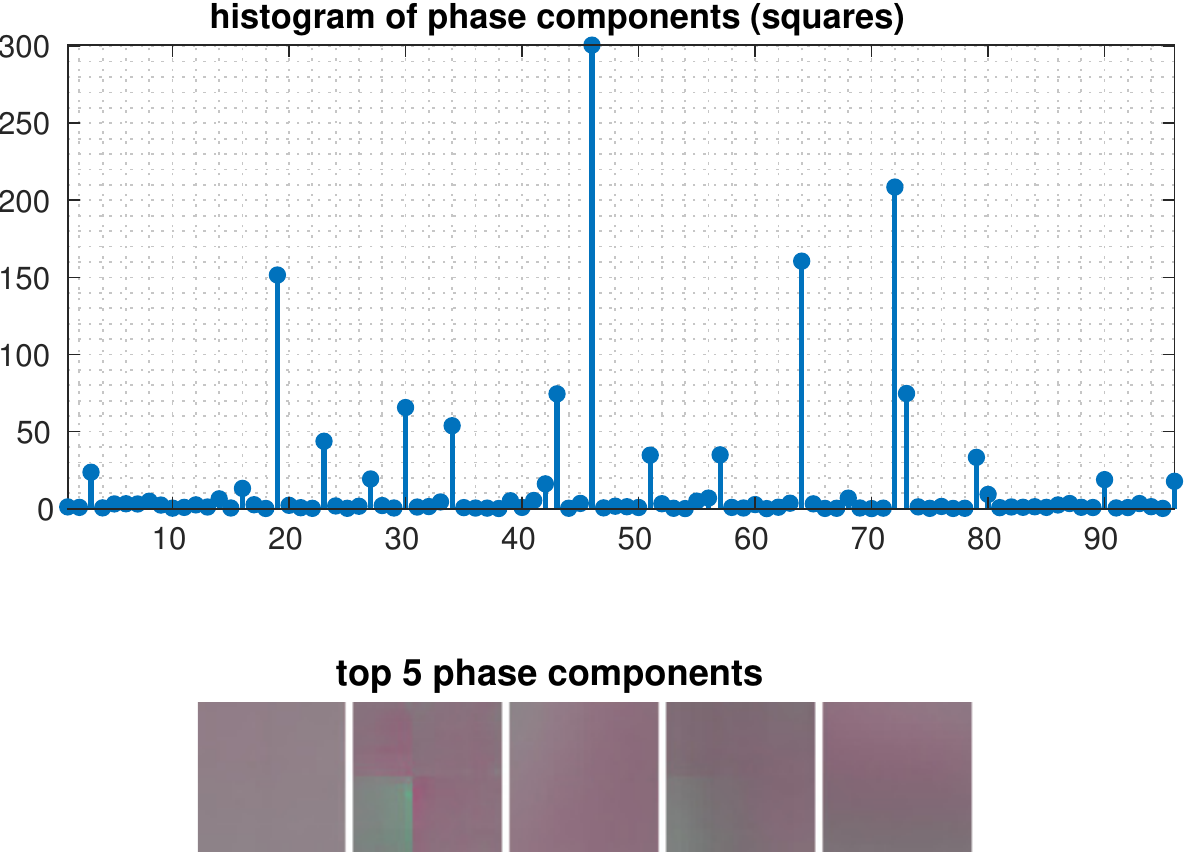}
	\caption{The elementary phase components for the study dataset, which are produced by the proposed method.}
	\label{fg:phs_comp_4cpx}
\end{figure}

One of the advantages of the proposed method than the nonlocal-based methods is that it can provide an insight into the elementary phase components for the study dataset. As demonstrated in \Fig \ref{fg:phs_comp_4cpx}, we calculate the summation of the amplitude values of the sparse coefficient maps, which present the contributions of the learned convolutional filters. Accordingly, the top five phase components are shown in \Fig \ref{fg:phs_comp_4cpx}. It can be obviously seen that different interferometric patterns have different codes of filters. For example, the most contributions of mountain and peaks are low-pass components, since the corresponding dominant filters are smooth. For the shear plane, the interferometric phase is mainly composed of both phase jumps and phase planes. As for squares, the filter with the rectangle shape becomes one of the dominant phase components.

\begin{figure}[!t]
	\centering
	\includegraphics[width=0.24\textwidth]{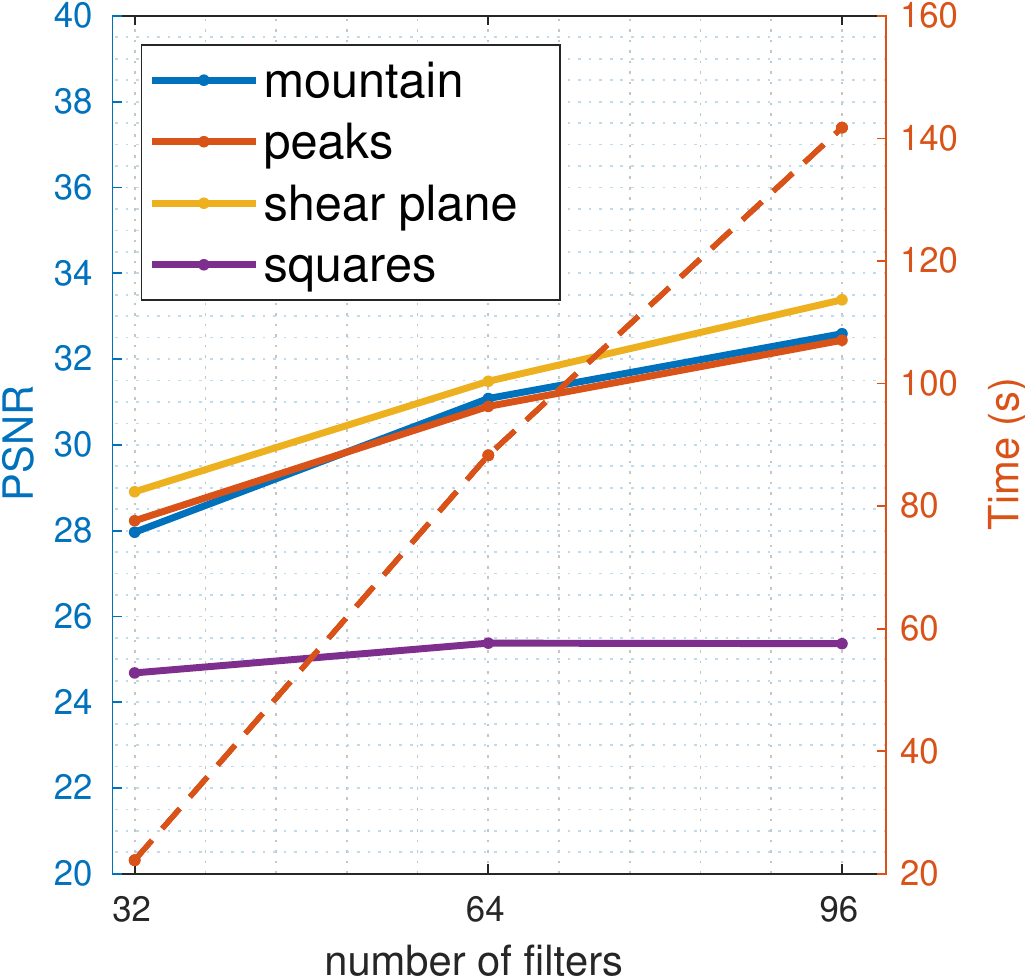}~
	\includegraphics[width=0.24\textwidth]{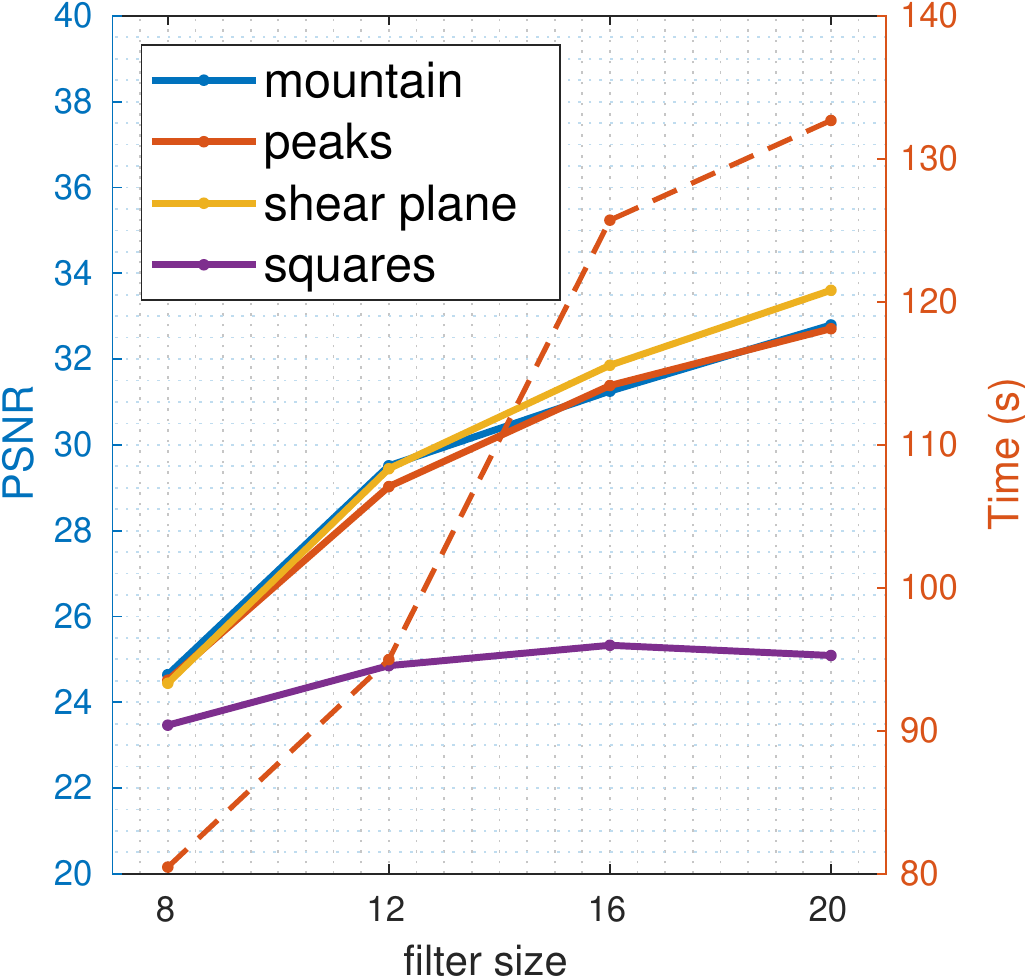}
	\caption{The efficiency study under different parameter settings, \textit{i.e.} the number of filters $ M $ and the filter size $ L $. The solid lines represent PSNR versus $M$ or $L$; the dashed lines represent time consumption versus $M$ or $L$.}
	\label{fg:effi_sty_filter_num_sz}
\end{figure}

There are mainly four parameters to be tuned in the proposed method, \textit{i.e.} $ \lambda $, $ \mu $, the number of filters $ M $  and the filter size $ L $. In our experiments, $ \lambda $ is set as $ 0.2 $ in the dictionary learning step, while $ 2.5 $ in the phase restoration step given the noisy phase input. $ \mu $ in ComCSC-GR is set in the range of $ [2,10] $ as the penalty parameter of the gradient regularization. Of course, the parameter setting is not limited to this, when a different dataset is processed. It is important to note here that the performances of phase reconstruction with respect to the other two parameters, \textit{i.e.} $ M $ and $ L $. Based on ComCSC-GR, we demonstrate the efficiency study of different parameter settings in \Fig \ref{fg:effi_sty_filter_num_sz}. It can been seen that larger number and larger size of filters can improve the reconstruction results of the interferometric phases. The plausible reason is that the representation capability can be enhanced as the number of parameters to be learned increases.  Moreover, not surprisingly, larger $M$ and $L$ lead to higher computational time. \Fig \ref{fg:effi_sty_filter_num_sz} shows that the time consumption is nearly linear to $M$, which supports our theoretical complexity bound $\mathcal{O}(TMN\log(N))$ given in Section \ref{sec:csc_complexity}. The quantitative relationship between $L$ and computational cost is not obvious. In our algorithm, each filter $\mathbf{d}_m$ of size $L$ is zero-padded to size $N$ ($L<N$) and FFT is conducted on the padded kernel. Since there are many zeros in the padded kernel, FFT can be faster than the theoretical upper bound $\mathcal{O}(N\log(N))$. This is why the computational cost is related with $L$. However, it highly depends on the implementation rather than the algorithm itself. Thus we just give a worst-case upper bound in Section \ref{sec:csc_complexity}.

\subsection{Real Data}

\begin{figure}[!t]
	\centering
	\includegraphics[width=0.23\textwidth]{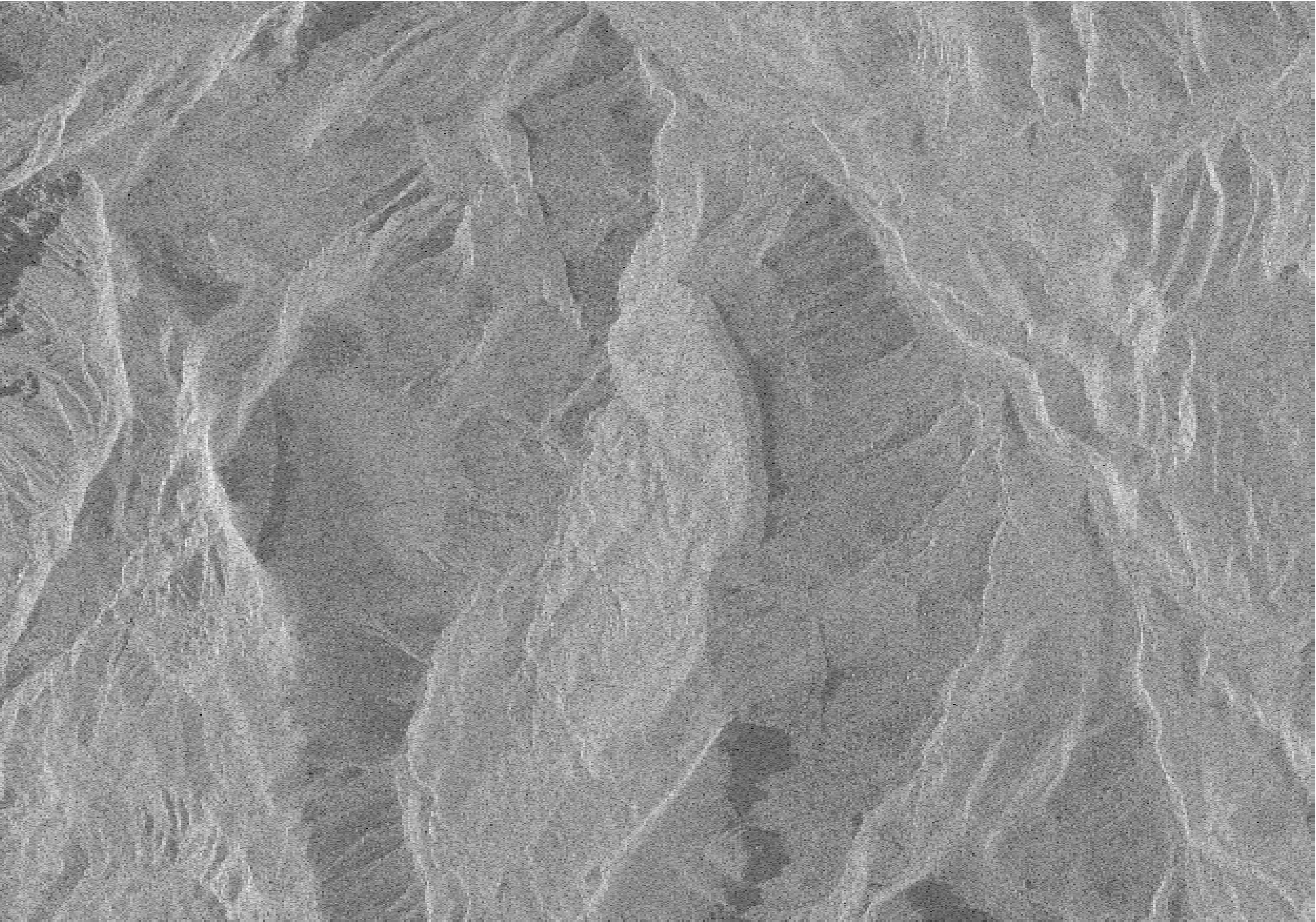}~
	\includegraphics[width=0.23\textwidth]{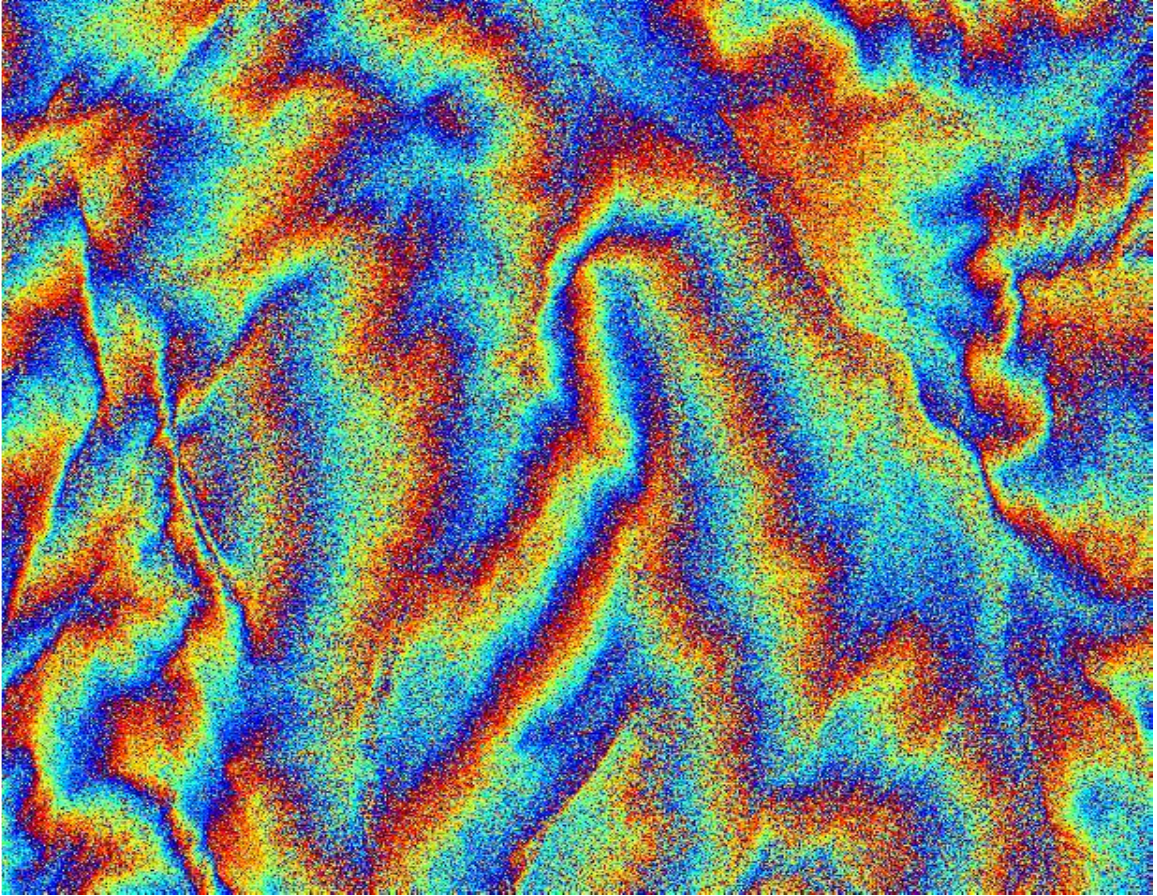}
	\caption{(Left) The mountainous study area of TerraSAR-X StripMap data shown by the amplitude (log scale). (Right) The corresponding unfiltered interferogram processed with ESA SNAP toolbox.}
	\label{fg:real_mountain_study_area}
\end{figure}

\begin{figure*}[!t]
	\centering
	\includegraphics[width=\textwidth]{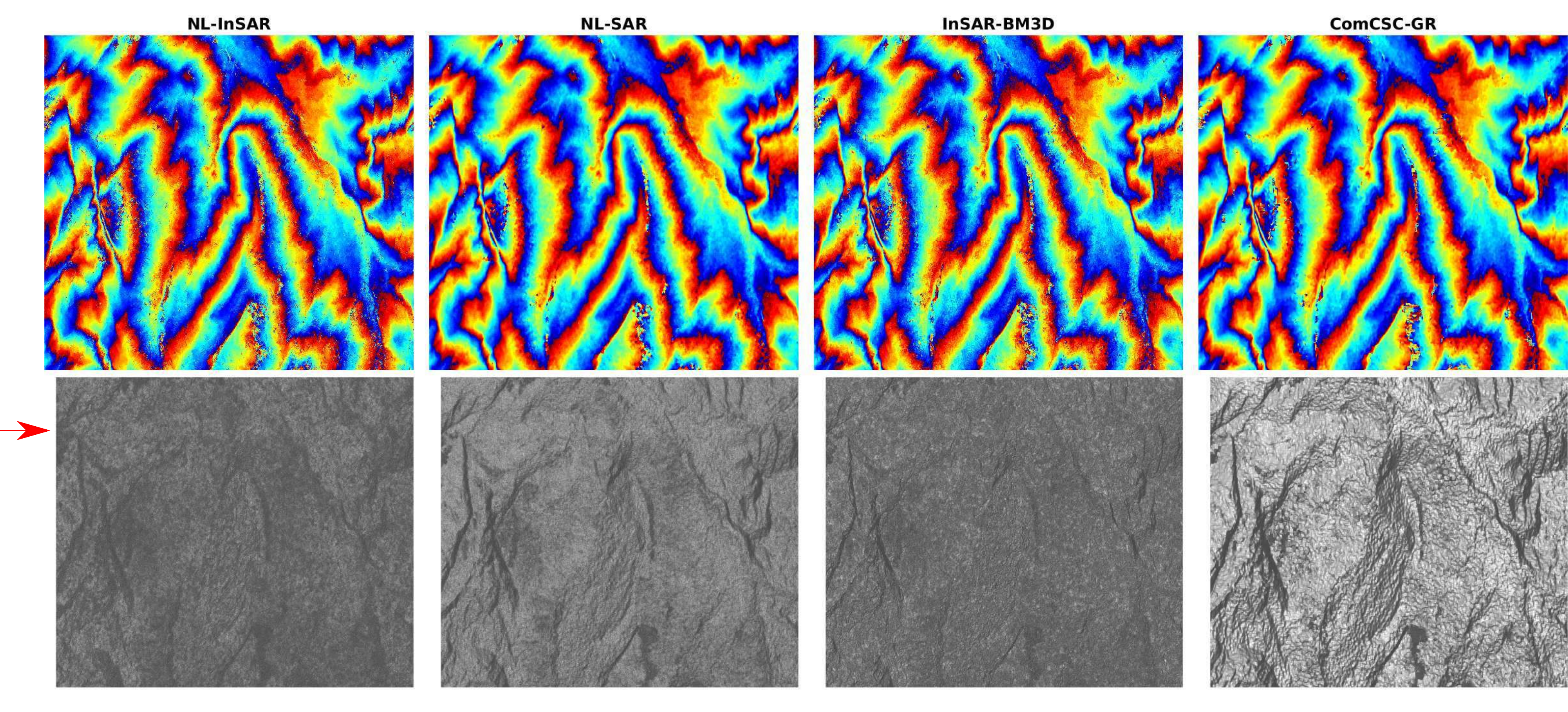}
	\caption{The filtered results of the four comparing methods (top row) and the corresponding unwrapped phases (bottom row). As notified by the red arrow, the profiles of the unwrapped phases are plotted in \Fig \ref{fg:extracted_one_profile} (Left) and one zoom-in area of the profile is illustrated to its right. }
	\label{fg:real_mountain_result}
\end{figure*}

\begin{figure}[!t]
	\centering
	\includegraphics[width=0.24\textwidth]{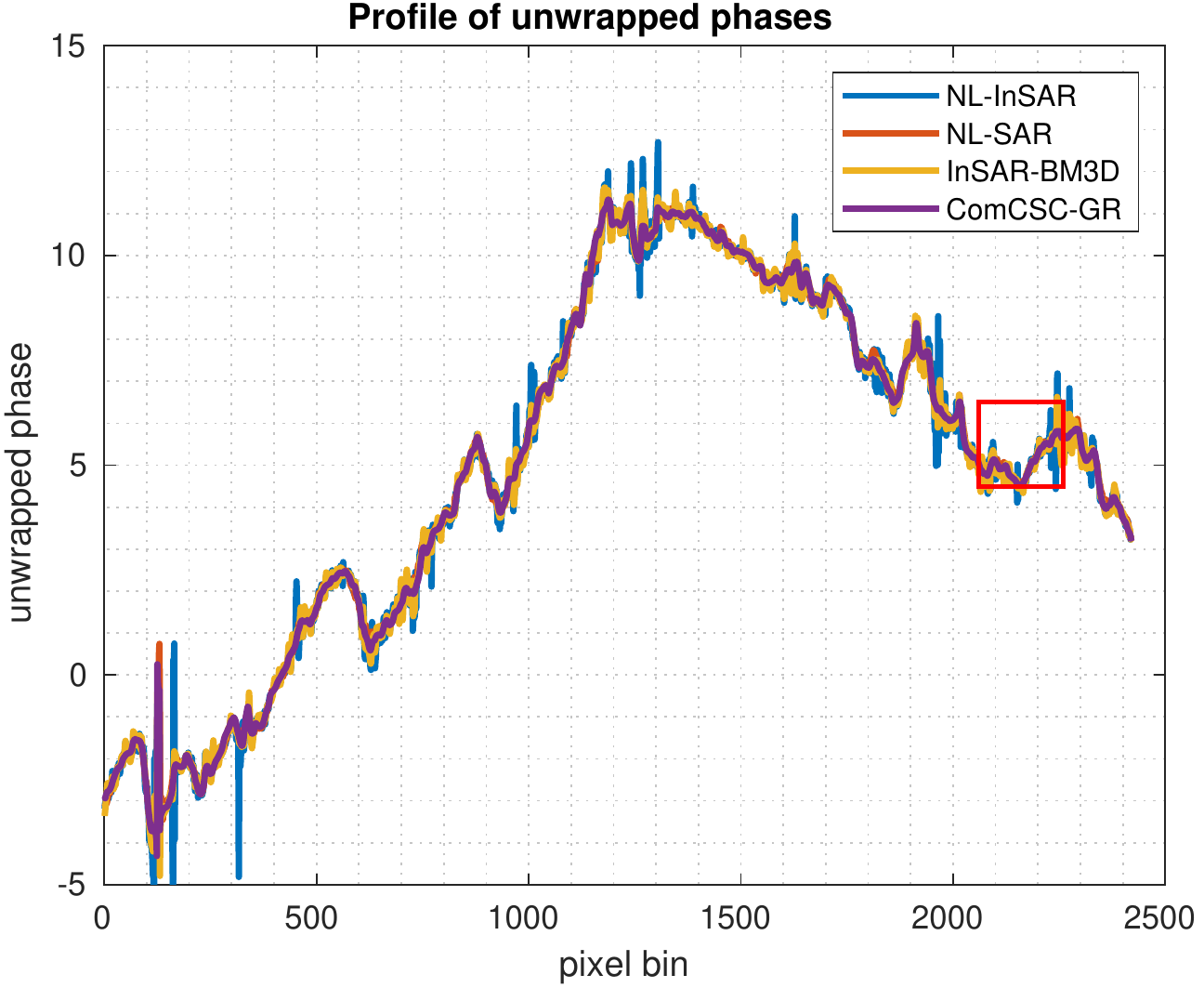}~
	\includegraphics[width=0.24\textwidth]{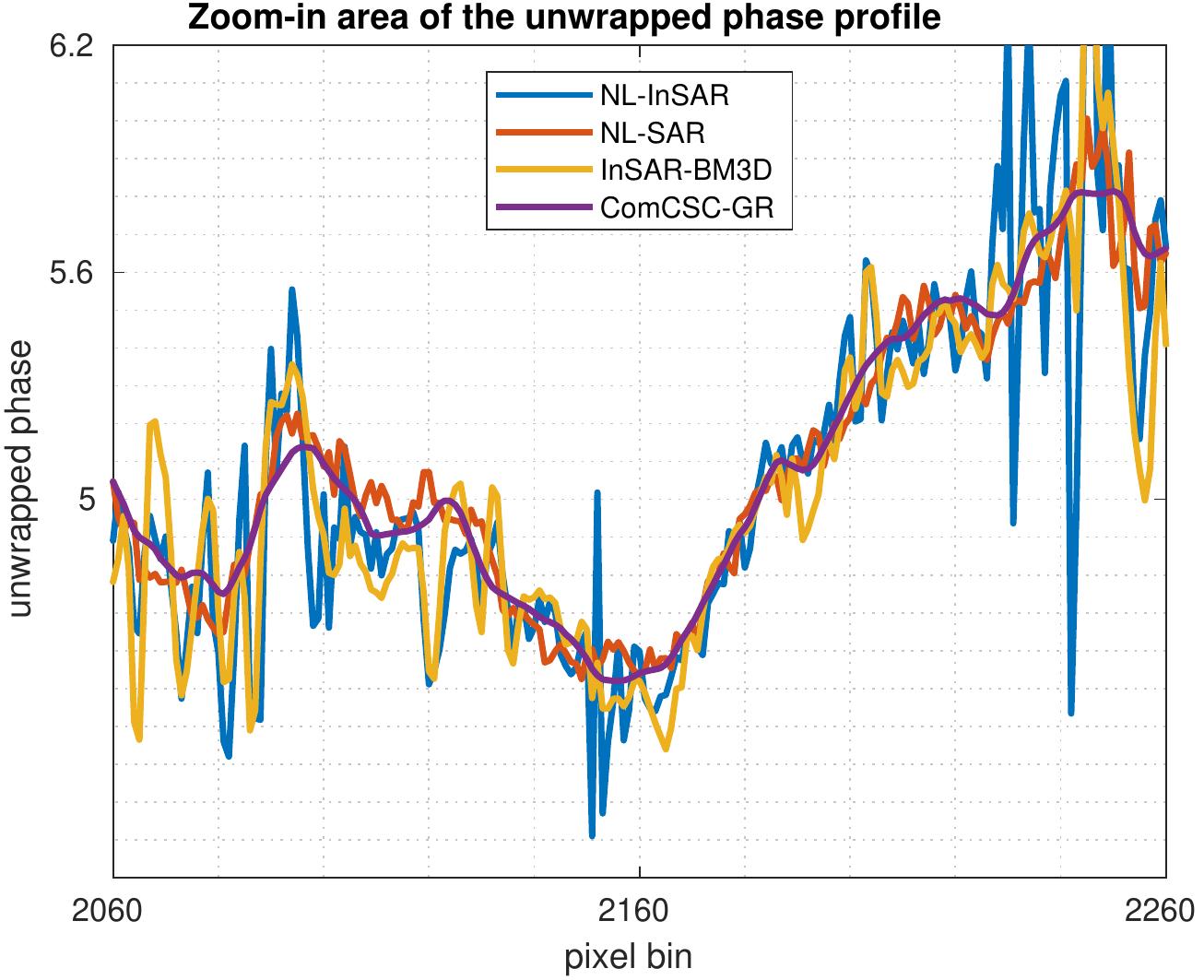}
	\caption{Profiles extracted on the unwrapped phases in \Fig \ref{fg:real_mountain_result} (indicated by the red arrow) and one zoom-in area. }
	\label{fg:extracted_one_profile}
\end{figure}

\begin{figure}[!t]
	\centering
	\includegraphics[width=0.45\textwidth]{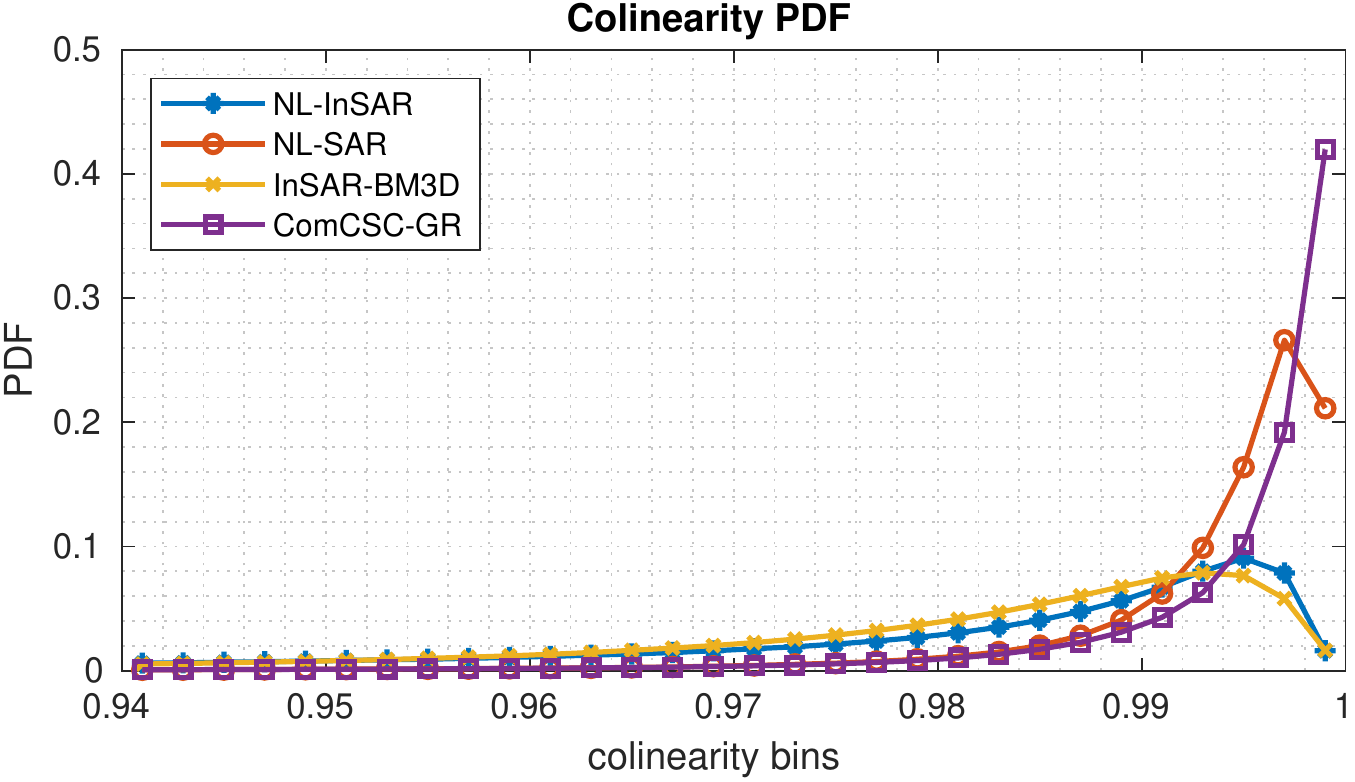}
	\caption{Colinearity probability density functions of the comparing methods, which are calculated based on the filtered interferograms.}
	\label{fg:moutain_colinearity_pdf}
\end{figure}

The first experiment is carried out on a TerraSAR-X StripMap dataset provided by AIRBUS Sample Imagery. The images are acquired on the are of Grand Canyon National Park, Arizona, USA, with the incidence angle of $ 39.2 $ and the time of $ 03/10/2008 $ and $ 03/21/2008 $. \Fig \ref{fg:real_mountain_study_area} shows the acquired mountainous area. The spatial size of this area is $ 1670\times2420 $.
The interferogram is processed with ESA SNAP toolbox and the original unfiltered interferogram is illustrated in its bottom.
Since there is no ground truth reference available, we first show the filtered results of the four comparing methods, \textit{i.e.} NL-InSAR, NL-SAR, InSAR-BM3D and ComCSC-GR, in the top row of \Fig \ref{fg:real_mountain_result} under the similar parameter setting as the simulations. By utilizing the same phase unwrapping algorithm, the bottom row of \Fig \ref{fg:real_mountain_result} demonstrates the corresponding unwrapped phases. Furthermore, the profiles of the unwrapped phases delineated by the red arrow are plotted in \Fig \ref{fg:extracted_one_profile} (Left) and one zoom-in area of the profile is illustrated to its right.

As shown in \Fig \ref{fg:real_mountain_result}, even though all the methods can greatly mitigate the noise, the result of NL-InSAR still contains noisy artifacts, compared with the other methods. Therefore, the corresponding topography revealed by the unwrapped phase is much noisier than the other methods.
NL-SAR and InSAR-BM3D achieve visually, appealing filtering performance, owing to the noise suppression and smoothness preservation.
However, as illustrated in the unwrapped phases, the topographic variations of the mountain are indicated more clearly in the proposed method than NL-SAR and InSAR-BM3D. Moreover, as observed in the extracted profiles of the unwrapped phase (\Fig \ref{fg:extracted_one_profile}), very sharp variations of the phases exist in the result of NL-InSAR and cannot correctly indicate the elevations of the mountain. Also, the staircase effect of NL-InSAR method is evidently observed in the zoom-in profile plot, especially in the decreasing and increasing slopes of the mountain. Consistent with the simulation, InSAR-BM3D cannot smoothly reconstruct the continuous variations of real phases. As illustrated in the zoom-in profile, the real phase vibration can be clearly observed and it will lead to the noisy DEM product. In comparison, the proposed method can greatly suppress the noise, maintain the details of phase fringes and avoid the staircase effect.

In order to evaluate the filtered interferograms without the high-resolution DEM ground truth of this area, we utilize the \textit{Colinearity Criterion} proposed in \cite{pinel2012multi}, which is defined as:
\begin{equation}
\begin{aligned}
C_i=\frac{|\sum_{p\in{M}_i}\exp(j(\phi_i-\phi_p))|}{M^2-1}\times\frac{\sum_{p\in{M}_i}|\exp(j(\phi_i-\phi_p))|}{M^2-1},
\end{aligned}
\end{equation}
where $ i $ is the index of the pixel to be assessed, $ p\in{M}_i $ denotes the close neighborhood pixels surrounding the $ i $th pixel, the local window size $ M $ is set as $ 7 $ (around $ 21\times21 [m^2] $ area), and $ \phi $ represents the real interferometric phase. This criterion measures the similarity of the phase history of the study pixel with respect to its surrounding ones. It is a measurement of homogeneity or smoothness given the neighboring pixels of interferometric phases. Higher values indicate the better homogeneity of the filtered InSAR phases. To some extent, it can be regarded as a quality assessment for the filtered interferogram, especially for the areas with homogeneous geophysical parameters. For this mountainous area, the main interferometric phase contribution is from the topography of this area. Within the area of $ 21\times21 [m^2] $, the colinearity should achieve high values, since the elevations of this small area are homogeneous. As displayed in \Fig \ref{fg:moutain_colinearity_pdf}, most points can achieve a colinearity above $ 0.9 $ in all the methods. However, the best homogeneity of
the filtered interferogram can be obtained by the proposed method, which can indicate that the underlying continuous variation of the elevation in this area can be smoothly reconstructed by ComCSC-GR.

\begin{figure}[!t]
	\centering
	\includegraphics[width=0.23\textwidth]{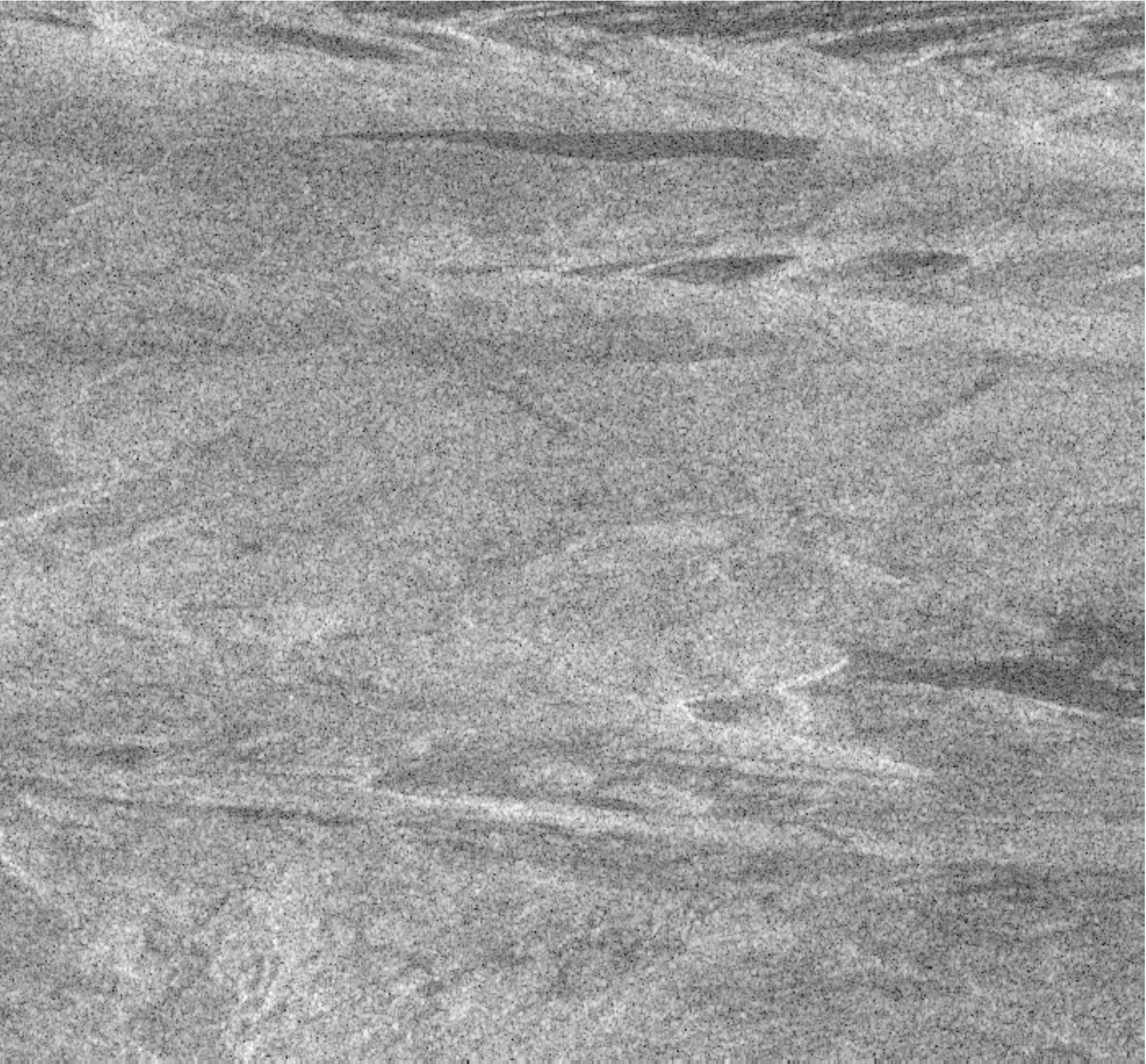}~
	\includegraphics[width=0.23\textwidth]{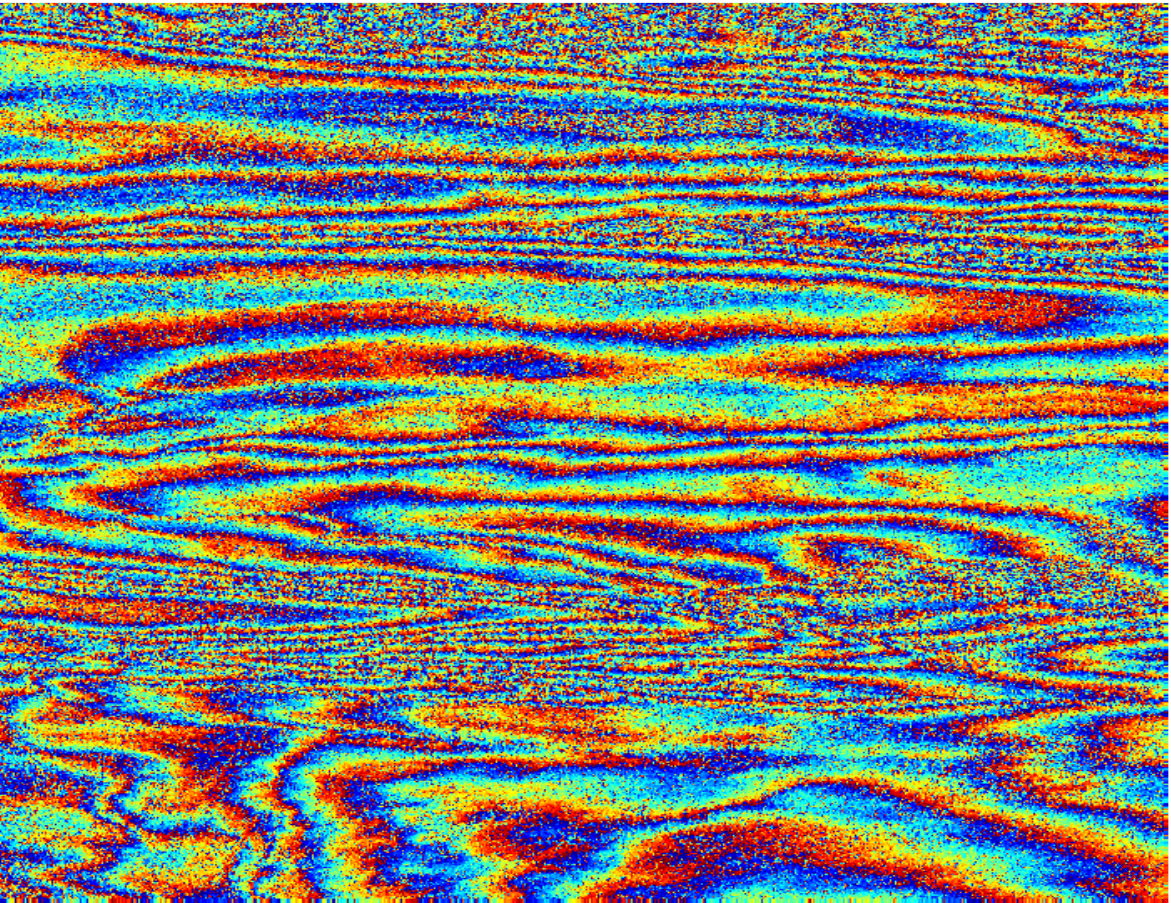}
	\caption{(Left) The mountainous study area of Sentinel-1 IW data shown by the amplitude (log scale). (Right) The corresponding unfiltered interferogram processed with ESA SNAP toolbox.}
	\label{fg:real_Sen1_mountain_study_area}
\end{figure}

\begin{figure*}[!t]
	\centering
	\includegraphics[width=\textwidth]{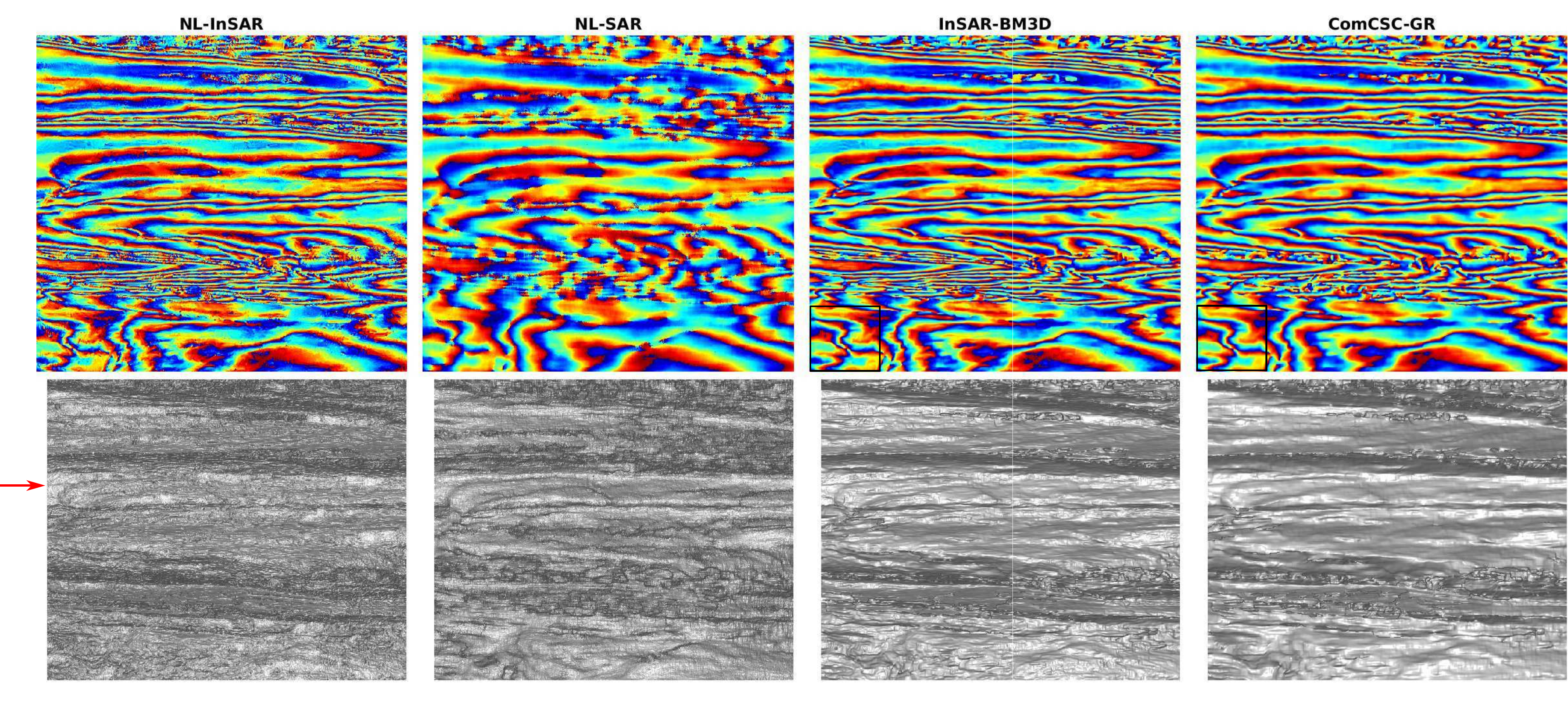}
	\caption{The filtered results of the four comparing methods (top row) and the corresponding unwrapped phases (bottom row). As notified by the red arrow, the profiles of the unwrapped phases are plotted in \Fig \ref{fig:real_Sen1_phase_profile} and one zoom-in area indicated by the black box at the bottom-left corner of InSAR-BM3D and ComCSC-GR results is illustrated in \Fig \ref{fig:zoomin_Sen1_area}. }
	\label{fg:real_Sen1_mountain_result}
\end{figure*}

\begin{figure}[!t]
    \centering
    \includegraphics[width=0.45\textwidth]{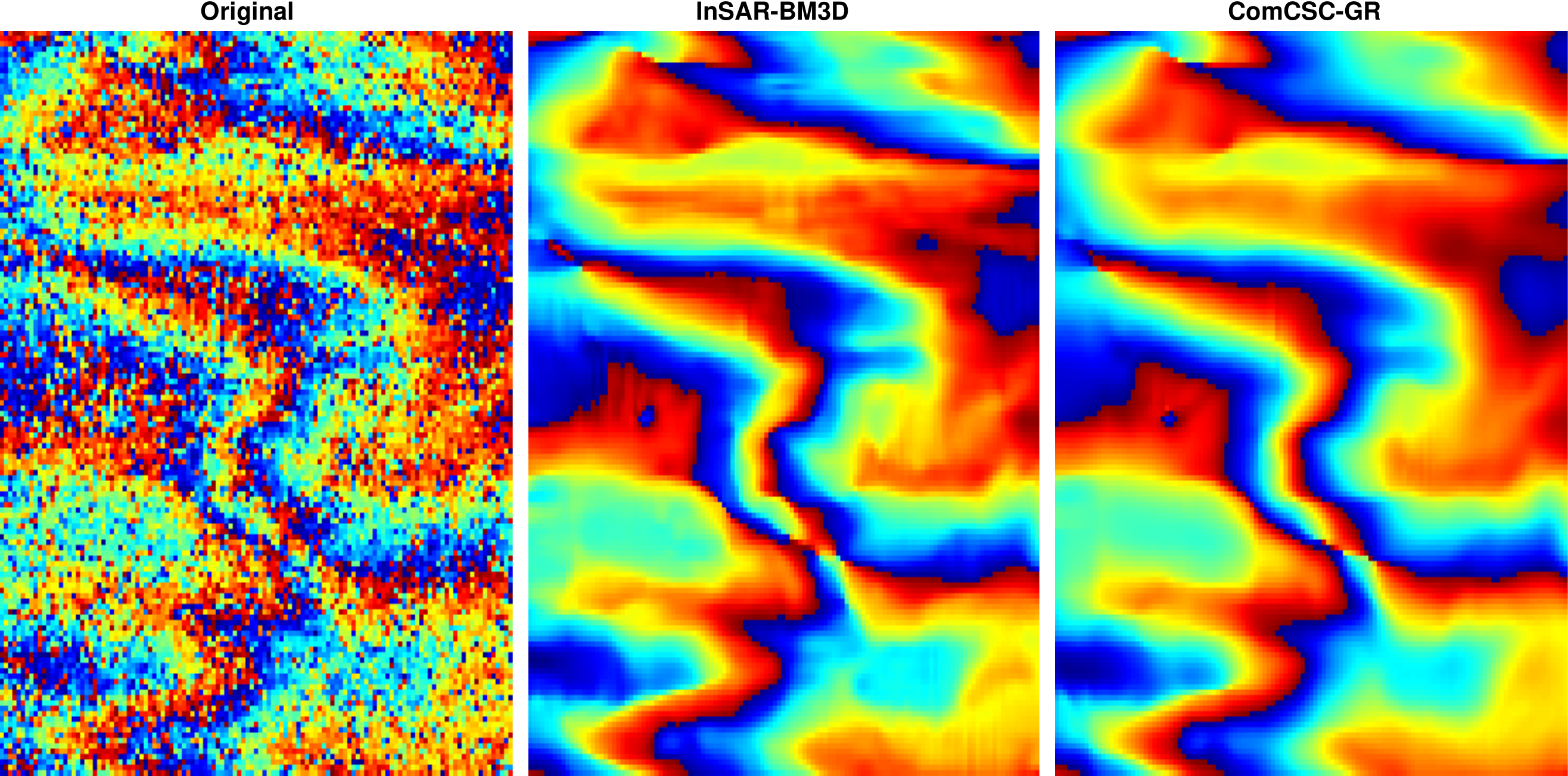}
    \caption{The comparison between the zoom-in area of the results from InSAR-BM3D and the proposed method (the black box area indicated in \Fig \ref{fg:real_Sen1_mountain_result}).}
    \label{fig:zoomin_Sen1_area}
\end{figure}

\begin{figure}[!t]
    \centering
    \includegraphics[width=0.45\textwidth]{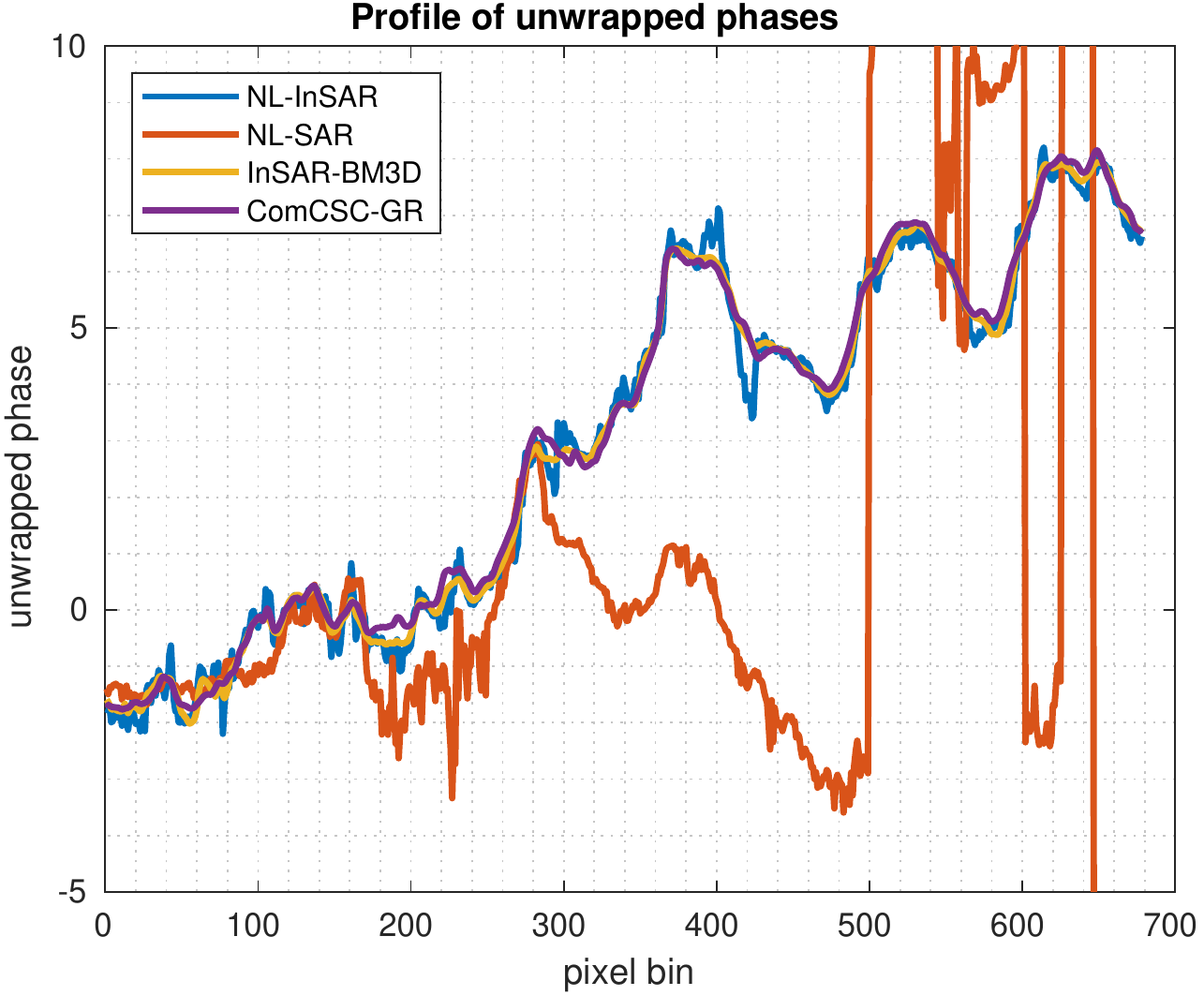}
    \caption{Profiles extracted in the results of unwrapped phases \Fig \ref{fg:real_Sen1_mountain_result} (indicated by the red arrow). }
    \label{fig:real_Sen1_phase_profile}
\end{figure}

\begin{figure}[!t]
    \centering
    \includegraphics[width=0.45\textwidth]{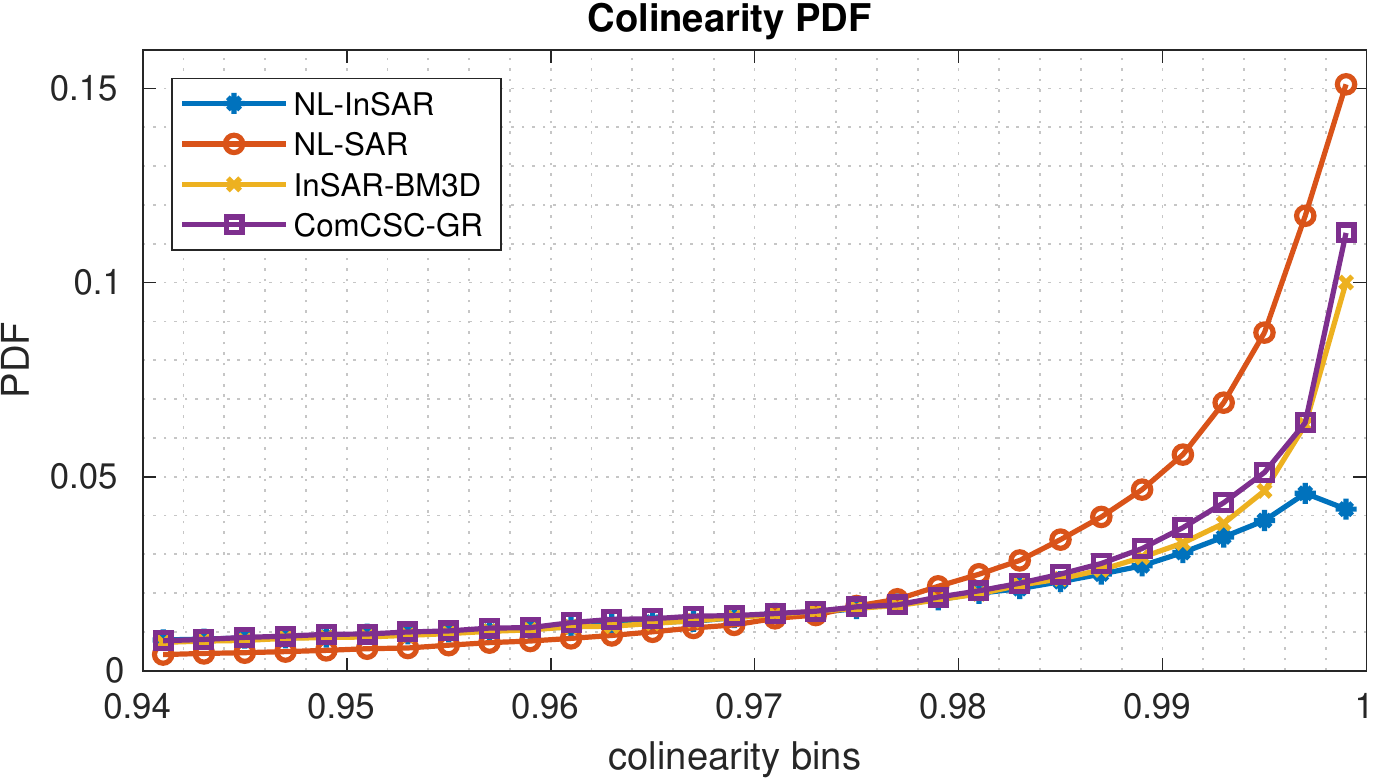}
    \caption{Colinearity probability density functions of the comparing methods, which are calculated based on the filtered interferograms.}
    \label{fig:Sen1_colinearity}
\end{figure}

The second experiment is conducted on another dataset of Alps mountains acquired by Sentinel-1 with the Interferometric Wide (IW) swath mode. The acquisition dates of the two images are $ 09/10/2018 $ and $ 15/10/2018 $, respectively. \Fig \ref{fg:real_Sen1_mountain_study_area} displays the amplitude SAR image and its original interferogram. The spatial size is of $ 650\times700 $ pixels. Compared with the first area, this interferogram is much denser and more heterogeneous. Similar with the above experiment, we filter the interferogram based on the four methods, \textit{i.e.} NL-InSAR, NL-SAR, InSAR-BM3D and ComCSC-GR and demonstrate the filtered results in \Fig \ref{fg:real_Sen1_mountain_result} (Top). At its bottom, the unwrapped phases are also calculated. One unwrapped phase profile is extracted to further compare the visualized performance among all the methods \Fig \ref{fig:real_Sen1_phase_profile}. Moreover, one zoom-in area at the bottom-left of the filtered results based on InSAR-BM3D and the proposed method is cropped and displayed in \Fig \ref{fig:zoomin_Sen1_area}.

As shown in \Fig \ref{fg:real_Sen1_mountain_result}, for this dense fringe area, NL-SAR result is over-smoothed, especially in the heterogeneous areas. For example, at the bottom part of the result, the fringes filtered by NL-SAR are blurred and cannot reflect the correct topography. Additionally, due to the over-smoothing issue, the unwrapped phase of the extracted profile is also incorrect (\Fig \ref{fig:real_Sen1_phase_profile}). As a comparison, based on the learned high-frequency convolutional kernels, the proposed method can very well preserve such heterogeneous characteristic of the dataset and mitigate the noise simultaneously. Although NL-InSAR can also preserve the phase details, the staircase effect in the restoration seriously influences the quality of the result, as illustrated by the unwrapped phases in \Fig \ref{fg:real_Sen1_mountain_result} and the profiles \Fig \ref{fig:real_Sen1_phase_profile}. Both InSAR-BM3D and ComCSC-GR can very well restore the underlying phases for this area with the dense fringe preservation and the noise mitigation. However, as the results shown in the zoom-in area \Fig \ref{fig:zoomin_Sen1_area}, InSAR-BM3D cannot efficiently reconstruct the phases near the fringe edges. It can be clearly observed that some vertical or horizontal artifacts exist in the filtered interferograms. In a contrast, not only the details of the fringe edges can be protected, but also the regions nearby can be more consistently reconstructed. As demonstrated the colinearity assessment in \Fig \ref{fig:Sen1_colinearity}, due to the over-smoothing effect of NL-SAR, it achieves the highest scores of the colinearity. However, NL-SAR loses the structural information of the topography in this area. Slightly better colinearity performance can be reached by the proposed method than InSAR-BM3D. Due to the staircase effect of NL-InSAR, the similarities among the pixels of the neighborhood cannot be high.

\section{Conclusion}\label{sc:conclustion}
In this paper, the convolutional sparse coding algorithm and its gradient regularized version are proposed in the complex domain. They can be exploited for interferometric phase restoration, which can avoid the staircase effect and preserve the details of phase variations. Moreover, ComCSC can decompose the global image in a deconvolutional manner, which can provide an insight for the elementary phase components for the interferometric phases. The corresponding performance is validated on both the synthetic and realistic high- and medium-resolution datasets from  TerraSAR-X StripMap and Sentinel-1 IW mode, respectively, with the comparison of the other state-of-the-art methods.

As indicated by the experiments, more accurate phase restoration can be achieved as the filter number and filter size increase. However, considering the computational costs, those two parameters cannot be set too big, since more parameters should be learned with the support of more training data. In our experiments, appealing results of both simulations and real datasets can be obtained with the $ 96 $ filters and $ 20\times20 $ filter size. Moreover, theoretically, the proposed method can be directly operated on the whole interferogram rather than patch-wisely as the conventional dictionary learning. In practice, when the spatial dimension of the studied dataset is too large for processing, the algorithm can also be carried out in a sliding-window manner.

As a future work, we plan to investigate the extension of the algorithm on InSAR stacks for simultaneously denoising multi-temporal interferograms. Also, in order to improve the performance of the sparse coding step, learning convolutional filters with multiple sizes can be further researched.



%

\appendices

\section{Discussion of (\ref{eq:convsparsecoding_cpx_softthres}) }
\label{app:proof}
In this section, we want to prove that the $\mathbf{Y}$ given by (\ref{eq:convsparsecoding_cpx_softthres}) is a solution of problem (\ref{eq:convsparsecodeing_Y_subproblem}).

Before the main proof, we start with a simplier case.
\begin{lemma}
\label{lemma:l1_sln}
The minimizer of the following problem
 \begin{equation}
\label{eq:l1_general}
    \min_{y}
	 \gamma  | y|+\frac{1}{2}|y-z|^2,
\end{equation}
where $y,z \in \mathbb{C}$ are both complex numbers, is
\begin{equation}
    \label{eq:l1_sln}
    y = \frac{z}{|z|} \cdot \max(0,|z|-\gamma).
\end{equation}
\end{lemma}
\begin{proof}
 First we use the polar form to represent complex numbers $y,z$: \[y  = \rho e^{i \theta } , z = \rho_0 e^{i \theta_0 },~~\rho,\rho_0 \geq 0,~~\theta,\theta_0 \in \mathbb{R}, \] where $\rho,\rho_0$ are the amplitude values and $ \theta,\theta_0$ are the phases. Then (\ref{eq:l1_general}) can be written as
\[
\begin{aligned}
    & \min_{y} \gamma| y|+\frac{1}{2}|y-z|^2 \\ =& \min_{\rho,\theta} \gamma  \rho + \frac{1}{2} \Big |\rho e^{i \theta } - \rho_0 e^{i \theta_0 }\Big |^2\\
    =& \min_{\rho\geq 0,\theta}\gamma \rho + \frac{1}{2} \Big (\rho e^{-i \theta } - \rho_0 e^{-i \theta_0 }\Big )\Big (\rho e^{i \theta } - \rho_0 e^{i \theta_0 }\Big )\\
    =& \min_{\rho\geq 0,\theta}\gamma \rho + \frac{1}{2} \Big (\rho^2 + \rho_0^2 -2\rho \rho_0 \cos{(\theta-\theta_0)}\Big )\\
    =& \min_{\rho\geq 0}\gamma \rho + \frac{1}{2}(\rho-\rho_0)^2 ~~(\text{real number soft-thresholding})
\end{aligned}
\]
The above equations show that the optimal $\theta = \theta_0$. Furthermore, since $\argmin_{\rho \geq 0} \gamma\rho + (\rho-\rho_0)^2/2 = \max(0,\rho_0-\gamma )$, we have the solution of (\ref{eq:l1_general}):
\[y = \max(0,\rho_0-\gamma ) e^{i\theta_0}\]In another word, $y = \frac{z}{|z|} \cdot \max(0,|z|-\gamma)$, (\ref{eq:l1_sln}) is proved.
\end{proof}

Given Lemma \ref{lemma:l1_sln} in hand, we can solve (\ref{eq:convsparsecodeing_Y_subproblem}) now.

Recall problem (\ref{eq:convsparsecodeing_Y_subproblem}):
\[
	\min_\mathbf{Y}\lambda\|\mathbf{Y}\|_{1,1}+\frac{\rho}{2}\|\mathbf{X}-\mathbf{Y}+\mathbf{U}\|_F^2.
\]
By definition (\ref{eq:concate}), tensors  $\mathbf{X,Y,U}$ are all obtained by concatenating several vectors. Based on the definition of norms $\|\cdot \|_{1,1}$ and $\|\cdot \|_F$, problem  (\ref{eq:convsparsecodeing_Y_subproblem}) is equivalent with
\[
	\min_{\{\mathbf{y}_{m,k}\}}\sum_{m=1}^M\sum_{k=1}^K \bigg( \lambda\|\mathbf{y}_{m,k}\|_{1}+\frac{\rho}{2}\|\mathbf{x}_{m,k}-\mathbf{y}_{m,k}+\mathbf{u}_{m,k}\|_2^2 \bigg),
\]
where $\mathbf{x}_{m,k},\mathbf{y}_{m,k},\mathbf{u}_{m,k}$ are $N$ dimensional complex-valued vectors. By the definition of norms $\|\cdot\|_1$ and $\|\cdot\|_2$, the above problem can be further decomposed as
\[
	\min_{\{y_{m,k,i}\}} \sum_{m=1}^M\sum_{k=1}^K \sum_{i=1}^N \bigg(
	 \lambda| y_{m,k,i}|+\frac{\rho}{2}|x_{m,k,i}-y_{m,k,i}+u_{m,k,i}|^2 \bigg),
\]where $x_{m,k,i},y_{m,k,i},u_{m,k,i}$ are the $i$-th element of vectors $\mathbf{x}_{m,k},\mathbf{y}_{m,k},\mathbf{u}_{m,k}$ respectively.
Since the function to minimize is totally separable among $m,k,i$, solving the above minimization problem is equivalent with solving the following problem independently for $1\leq m \leq M,1\leq k \leq K,1\leq i \leq N$:
\begin{equation}
\label{eq:l1_decompose}
    \min_{\{y_{m,k,i}\}}
	 \lambda| y_{m,k,i}|+\frac{\rho}{2}|x_{m,k,i}-y_{m,k,i}+u_{m,k,i}|^2.
\end{equation}

Plugging
\[
\begin{aligned}
\gamma :=& \lambda / \rho\\
y := & y_{m,k,i}\\
z := & x_{m,k,i} + u_{m,k,i}
\end{aligned}
\]
into (\ref{eq:l1_general}), we can recover (\ref{eq:l1_decompose}). Thus, the solution of (\ref{eq:l1_decompose}) is
\[y_{m,k,i} = \frac{x_{m,k,i} + u_{m,k,i}}{|x_{m,k,i} + u_{m,k,i}|} \cdot \max(0,|x_{m,k,i} + u_{m,k,i}|-\lambda/\rho ).\]
Concatenating the above equation for $1\leq m \leq M,1\leq k \leq K,1\leq i \leq N$, we obtain
\[ \mathbf{Y} = \frac{\mathbf{X+U}}{|\mathbf{X+U}|} \odot \max\Big (0, |\mathbf{X+U}| - \lambda/\rho\Big ),\]where all the operations are conducted elementwisely. That is exactly
$
\mathbf{Y}=\mathcal{CS}_{\frac{\lambda}{\rho}}\left(\mathbf{X}+\mathbf{U}\right)
$. We proved that the $\mathbf{Y}$ given by (\ref{eq:convsparsecoding_cpx_softthres}) is a solution of problem (\ref{eq:convsparsecodeing_Y_subproblem}).

%
%
%
%




\bibliographystyle{IEEEtran}
\bibliography{reference}

\begin{thebibliography}{10}
\providecommand{\url}[1]{#1}
\csname url@samestyle\endcsname
\providecommand{\newblock}{\relax}
\providecommand{\bibinfo}[2]{#2}
\providecommand{\BIBentrySTDinterwordspacing}{\spaceskip=0pt\relax}
\providecommand{\BIBentryALTinterwordstretchfactor}{4}
\providecommand{\BIBentryALTinterwordspacing}{\spaceskip=\fontdimen2\font plus
\BIBentryALTinterwordstretchfactor\fontdimen3\font minus
  \fontdimen4\font\relax}
\providecommand{\BIBforeignlanguage}[2]{{%
\expandafter\ifx\csname l@#1\endcsname\relax
\typeout{** WARNING: IEEEtran.bst: No hyphenation pattern has been}%
\typeout{** loaded for the language `#1'. Using the pattern for}%
\typeout{** the default language instead.}%
\else
\language=\csname l@#1\endcsname
\fi
#2}}
\providecommand{\BIBdecl}{\relax}
\BIBdecl

\bibitem{lee1983simple}
J.-S. Lee, ``A simple speckle smoothing algorithm for synthetic aperture radar
  images,'' \emph{IEEE Trans. Syst. Man Cybern.}, no.~1, pp. 85--89, 1983.

\bibitem{lee1998insarfilter}
{Jong-Sen Lee}, K.~P. {Papathanassiou}, T.~L. {Ainsworth}, M.~R. {Grunes}, and
  A.~{Reigber}, ``A new technique for noise filtering of sar interferometric
  phase images,'' \emph{IEEE Trans. Geosci. Remote Sens.}, vol.~36, no.~5, pp.
  1456--1465, 1998.

\bibitem{goldstein1998radar}
R.~M. Goldstein and C.~L. Werner, ``Radar interferogram filtering for
  geophysical applications,'' \emph{Geophys. Res. Lett.}, vol.~25, no.~21, pp.
  4035--4038, 1998.

\bibitem{buades2005non}
A.~Buades, B.~Coll, and J.-M. Morel, ``A non-local algorithm for image
  denoising,'' in \emph{Proc. CVPR}, vol.~2.\hskip 1em plus 0.5em minus
  0.4em\relax IEEE, 2005, pp. 60--65.

\bibitem{deledalle2009iterative}
C.-A. Deledalle, L.~Denis, and F.~Tupin, ``Iterative weighted maximum
  likelihood denoising with probabilistic patch-based weights,'' \emph{IEEE
  Trans. Image Process.}, vol.~18, no.~12, pp. 2661--2672, 2009.

\bibitem{parrilli2011nonlocal}
S.~Parrilli, M.~Poderico, C.~V. Angelino, and L.~Verdoliva, ``A nonlocal sar
  image denoising algorithm based on llmmse wavelet shrinkage,'' \emph{IEEE
  Trans. Geosci. Remote Sens.}, vol.~50, no.~2, pp. 606--616, 2011.

\bibitem{cozzolino2013fast}
D.~Cozzolino, S.~Parrilli, G.~Scarpa, G.~Poggi, and L.~Verdoliva, ``Fast
  adaptive nonlocal sar despeckling,'' \emph{IEEE Geosci. Remote Sens. Lett.},
  vol.~11, no.~2, pp. 524--528, 2013.

\bibitem{di2016scattering}
G.~Di~Martino, A.~Di~Simone, A.~Iodice, and D.~Riccio, ``Scattering-based
  nonlocal means sar despeckling,'' \emph{IEEE Trans. Geosci. Remote Sens.},
  vol.~54, no.~6, pp. 3574--3588, 2016.

\bibitem{deledalle2011nl}
C.-A. Deledalle, L.~Denis, and F.~Tupin, ``Nl-insar: Nonlocal interferogram
  estimation,'' \emph{IEEE Trans. Geosci. Remote Sens.}, vol.~49, no.~4, pp.
  1441--1452, 2011.

\bibitem{deledalle2015nl}
C.-A. Deledalle, L.~Denis, F.~Tupin, A.~Reigber, and M.~J{\"a}ger, ``Nl-sar: A
  unified nonlocal framework for resolution-preserving (pol)(in) sar
  denoising,'' \emph{IEEE Trans. Geosci. Remote Sens.}, vol.~53, no.~4, pp.
  2021--2038, 2015.

\bibitem{chen2013interferometric}
R.~Chen, W.~Yu, R.~Wang, G.~Liu, and Y.~Shao, ``Interferometric phase denoising
  by pyramid nonlocal means filter,'' \emph{IEEE Geosci. Remote Sens.Lett.},
  vol.~10, no.~4, pp. 826--830, 2013.

\bibitem{lin2014nonlocal}
X.~Lin, F.~Li, D.~Meng, D.~Hu, and C.~Ding, ``Nonlocal sar interferometric
  phase filtering through higher order singular value decomposition,''
  \emph{IEEE Geosci. Remote Sens. Lett.}, vol.~12, no.~4, pp. 806--810, 2014.

\bibitem{sica2018insar}
F.~Sica, D.~Cozzolino, X.~X. Zhu, L.~Verdoliva, and G.~Poggi, ``Insar-bm3d: A
  nonlocal filter for sar interferometric phase restoration,'' \emph{IEEE
  Trans. Geosci. Remote Sens.}, vol.~56, no.~6, pp. 3456--3467, 2018.

\bibitem{gao2019novel}
Y.~Gao, S.~Zhang, T.~Li, L.~Guo, Q.~Chen, and S.~Li, ``A novel two-step noise
  reduction approach for interferometric phase images,'' \emph{Opt. Laser.
  Eng.}, vol. 121, pp. 1--10, 2019.

\bibitem{chen2010nonlocal}
J.~Chen, Y.~Chen, W.~An, Y.~Cui, and J.~Yang, ``Nonlocal filtering for
  polarimetric sar data: A pretest approach,'' \emph{IEEE Trans. Geosci. Remote
  Sens.}, vol.~49, no.~5, pp. 1744--1754, 2010.

\bibitem{d2013iterative}
O.~D'Hondt, S.~Guillaso, and O.~Hellwich, ``Iterative bilateral filtering of
  polarimetric sar data,'' \emph{IEEE J. Sel. Top. Appl. Earth Obs. Remote
  Sens.}, vol.~6, no.~3, pp. 1628--1639, 2013.

\bibitem{sica2015nonlocal}
F.~Sica, D.~Reale, G.~Poggi, L.~Verdoliva, and G.~Fornaro, ``Nonlocal adaptive
  multilooking in sar multipass differential interferometry,'' \emph{IEEE J.
  Sel. Top. Appl. Earth Obs. Remote Sens.}, vol.~8, no.~4, pp. 1727--1742,
  2015.

\bibitem{d2017nonlocal}
O.~D’Hondt, C.~Lopez-Martinez, S.~Guillaso, and O.~Hellwich, ``Nonlocal
  filtering applied to 3-d reconstruction of tomographic sar data,'' \emph{IEEE
  Trans. Geosci. Remote Sens.}, vol.~56, no.~1, pp. 272--285, 2017.

\bibitem{hongxing2015interferometric}
H.~Hao, B.-D. Jos{\'e}~M, and K.~Vladimir, ``Interferometric phase image
  estimation via sparse coding in the complex domain,'' \emph{IEEE Trans.
  Geosci. Remote Sens.}, vol.~53, no.~5, pp. 2587--2602, 2015.

\bibitem{chen2001atomic}
S.~S. Chen, D.~L. Donoho, and M.~A. Saunders, ``Atomic decomposition by basis
  pursuit,'' \emph{SIAM Rev.}, vol.~43, no.~1, pp. 129--159, 2001.

\bibitem{bruckstein2009sparse}
A.~M. Bruckstein, D.~L. Donoho, and M.~Elad, ``From sparse solutions of systems
  of equations to sparse modeling of signals and images,'' \emph{SIAM Rev.},
  vol.~51, no.~1, pp. 34--81, 2009.

\bibitem{hong2015novel}
D.~Hong, W.~Liu, J.~Su, Z.~Pan, and G.~Wang, ``A novel hierarchical approach
  for multispectral palmprint recognition,'' \emph{Neurocomputing}, vol. 151,
  pp. 511--521, 2015.

\bibitem{kang2017robust}
J.~Kang, Y.~Wang, M.~K{\"o}rner, and X.~X. Zhu, ``Robust object-based multipass
  insar deformation reconstruction,'' \emph{IEEE Trans. Geosci. Remote Sens.},
  vol.~55, no.~8, pp. 4239--4251, 2017.

\bibitem{hong2018sulora}
D.~Hong and X.~X. Zhu, ``S{UL}o{RA}: Subspace unmixing with low-rank attribute
  embedding for hyperspectral data analysis,'' \emph{IEEE J. Sel. Top. Signal
  Process.}, vol.~12, no.~6, pp. 1351--1363, 2018.

\bibitem{kang2018object}
J.~Kang, Y.~Wang, M.~Schmitt, and X.~X. Zhu, ``Object-based multipass insar via
  robust low-rank tensor decomposition,'' \emph{IEEE Trans. Geosci. Remote
  Sens.}, vol.~56, no.~6, pp. 3062--3077, 2018.

\bibitem{hong2019augmented}
D.~Hong, N.~Yokoya, J.~Chanussot, and X.~X. Zhu, ``An augmented linear mixing
  model to address spectral variability for hyperspectral unmixing,''
  \emph{IEEE Trans. Image Process.}, vol.~28, no.~4, pp. 1923--1938, 2019.

\bibitem{huizhang2020}
H.~Yang, C.~Chen, S.~Chen, F.~Xi, and Z.~Liu, ``Non-common band sar
  interferometry via compressive sensing,'' \emph{IEEE Trans. Geosci. Remote
  Sens.}, 2020.

\bibitem{zhang2015fully}
L.~Zhang, L.~Sun, B.~Zou, and W.~M. Moon, ``Fully polarimetric sar image
  classification via sparse representation and polarimetric features,''
  \emph{IEEE J. Sel. Top. Appl. Earth Obs. Remote Sens.}, vol.~8, no.~8, pp.
  3923--3932, 2015.

\bibitem{yang2016riemannian}
W.~Yang, N.~Zhong, X.~Yang, and A.~Cherian, ``Riemannian sparse coding for
  classification of polsar images,'' in \emph{Proc. IGARSS}.\hskip 1em plus
  0.5em minus 0.4em\relax IEEE, 2016, pp. 5698--5701.

\bibitem{hong2017learning}
D.~Hong, N.~Yokoya, and X.~X. Zhu, ``Learning a robust local manifold
  representation for hyperspectral dimensionality reduction,'' \emph{IEEE J.
  Sel. Top. Appl. Earth Obs. Remote Sens.}, vol.~10, no.~6, pp. 2960--2975,
  2017.

\bibitem{zhong2017unsupervised}
N.~Zhong, W.~Yang, A.~Cherian, X.~Yang, G.-S. Xia, and M.~Liao, ``Unsupervised
  classification of polarimetric sar images via riemannian sparse coding,''
  \emph{IEEE Trans. Geosci. Remote Sens.}, vol.~55, no.~9, pp. 5381--5390,
  2017.

\bibitem{ren2018polsar}
B.~Ren, B.~Hou, Z.~Wen, W.~Xie, and L.~Jiao, ``Polsar image classification via
  multimodal sparse representation-based feature fusion,'' \emph{Int. J. Remote
  Sens.}, vol.~39, no.~22, pp. 7861--7880, 2018.

\bibitem{hang2019cascaded}
R.~Hang, Q.~Liu, D.~Hong, and P.~Ghamisi, ``Cascaded recurrent neural networks
  for hyperspectral image classification,'' \emph{IEEE Trans. Geosci. Remote
  Sens.}, vol.~57, no.~8, pp. 5384--5394, 2019.

\bibitem{xu2015iterative}
B.~Xu, Y.~Cui, Z.~Li, and J.~Yang, ``An iterative sar image filtering method
  using nonlocal sparse model,'' \emph{IEEE Geosci. Remote Sens. Lett.},
  vol.~12, no.~8, pp. 1635--1639, 2015.

\bibitem{liu2017sar}
S.~Liu, M.~Liu, P.~Li, J.~Zhao, Z.~Zhu, and X.~Wang, ``Sar image denoising via
  sparse representation in shearlet domain based on continuous cycle
  spinning,'' \emph{IEEE Trans. Geosci. Remote Sens.}, vol.~55, no.~5, pp.
  2985--2992, 2017.

\bibitem{liu2016image}
Y.~Liu, X.~Chen, R.~K. Ward, and Z.~J. Wang, ``Image fusion with convolutional
  sparse representation,'' \emph{IEEE signal processing letters}, vol.~23,
  no.~12, pp. 1882--1886, 2016.

\bibitem{hong2016robust}
D.~Hong, W.~Liu, X.~Wu, Z.~Pan, and J.~Su, ``Robust palmprint recognition based
  on the fast variation vese--osher model,'' \emph{Neurocomputing}, vol. 174,
  pp. 999--1012, 2016.

\bibitem{zeiler2010deconvolutional}
M.~D. Zeiler, D.~Krishnan, G.~W. Taylor, and R.~Fergus, ``Deconvolutional
  networks,'' in \emph{Proc. CVPR}, 2010.

\bibitem{chalasani2013fast}
R.~Chalasani, J.~C. Principe, and N.~Ramakrishnan, ``A fast proximal method for
  convolutional sparse coding,'' in \emph{Proc. IJCNN}.\hskip 1em plus 0.5em
  minus 0.4em\relax IEEE, 2013, pp. 1--5.

\bibitem{bristow2013fast}
H.~Bristow, A.~Eriksson, and S.~Lucey, ``Fast convolutional sparse coding,'' in
  \emph{Proc. CVPR}, 2013, pp. 391--398.

\bibitem{heide2015fast}
F.~Heide, W.~Heidrich, and G.~Wetzstein, ``Fast and flexible convolutional
  sparse coding,'' in \emph{Proc. CVPR}, 2015, pp. 5135--5143.

\bibitem{wohlberg2016efficient}
B.~Wohlberg, ``Efficient algorithms for convolutional sparse representations,''
  \emph{IEEE Trans. Image Process.}, vol.~25, no.~1, pp. 301--315, 2016.

\bibitem{liu2017online}
J.~Liu, C.~Garcia-Cardona, B.~Wohlberg, and W.~Yin, ``Online convolutional
  dictionary learning,'' in \emph{Proc. ICIP}.\hskip 1em plus 0.5em minus
  0.4em\relax IEEE, 2017, pp. 1707--1711.

\bibitem{liu2018first}
------, ``First-and second-order methods for online convolutional dictionary
  learning,'' \emph{SIAM J. Imaging Sci.}, vol.~11, no.~2, pp. 1589--1628,
  2018.

\bibitem{wu2019orsim}
X.~Wu, D.~Hong, J.~Tian, J.~Chanussot, W.~Li, and R.~Tao, ``Orsim detector: A
  novel object detection framework in optical remote sensing imagery using
  spatial-frequency channel features,'' \emph{IEEE Trans. Geosci. Remote
  Sens.}, vol.~57, no.~7, pp. 5146--5158, 2019.

\bibitem{wu2019fourier}
X.~Wu, D.~Hong, J.~Chanussot, Y.~Xu, R.~Tao, and Y.~Wang, ``Fourier-based
  rotation-invariant feature boosting: An efficient framework for geospatial
  object detection,'' \emph{IEEE Geosci. Remote Sens. Lett.}, vol.~17, no.~2,
  pp. 302--306, 2020.

\bibitem{gu2015convolutional}
S.~Gu, W.~Zuo, Q.~Xie, D.~Meng, X.~Feng, and L.~Zhang, ``Convolutional sparse
  coding for image super-resolution,'' in \emph{Proc. ICCV}, 2015, pp.
  1823--1831.

\bibitem{zhang2017convolutional}
H.~Zhang and V.~M. Patel, ``Convolutional sparse and low-rank coding-based
  image decomposition,'' \emph{IEEE Trans. Image Process.}, vol.~27, no.~5, pp.
  2121--2133, 2017.

\bibitem{zhang2017convolutional2}
------, ``Convolutional sparse and low-rank coding-based rain streak removal,''
  in \emph{Proc. WACV}.\hskip 1em plus 0.5em minus 0.4em\relax IEEE, 2017, pp.
  1259--1267.

\bibitem{boyd2011distributed}
S.~Boyd, N.~Parikh, E.~Chu, B.~Peleato, J.~Eckstein \emph{et~al.},
  ``Distributed optimization and statistical learning via the alternating
  direction method of multipliers,'' \emph{Found. Trends{\textregistered} Mach.
  Learn.}, vol.~3, no.~1, pp. 1--122, 2011.

\bibitem{hong2019cospace}
D.~Hong, N.~Yokoya, J.~Chanussot, and X.~X. Zhu, ``Co{S}pace: Common subspace
  learning from hyperspectral-multispectral correspondences,'' \emph{IEEE
  Trans. Geosci. Remote Sens.}, vol.~57, no.~7, pp. 4349--4359, 2019.

\bibitem{hong2019learnable}
D.~Hong, N.~Yokoya, N.~Ge, J.~Chanussot, and X.~X. Zhu, ``Learnable manifold
  alignment ({L}e{MA}): A semi-supervised cross-modality learning framework for
  land cover and land use classification,'' \emph{ISPRS J. Photogramm. Remote
  Sens.}, vol. 147, pp. 193--205, 2019.

\bibitem{hong2019learning}
D.~Hong, N.~Yokoya, J.~Chanussot, J.~Xu, and X.~X. Zhu, ``Learning to propagate
  labels on graphs: An iterative multitask regression framework for
  semi-supervised hyperspectral dimensionality reduction,'' \emph{ISPRS J.
  Photogramm. Remote Sens.}, vol. 158, pp. 35--49, 2019.

\bibitem{garcia2018convolutional}
C.~Garcia-Cardona and B.~Wohlberg, ``Convolutional dictionary learning: A
  comparative review and new algorithms,'' \emph{IEEE Trans. Comput. Imag.},
  vol.~4, no.~3, pp. 366--381, 2018.

\bibitem{hong2020invariant}
D.~Hong, X.~Wu, P.~Ghamisi, J.~Chanussot, N.~Yokoya, and X.~X. Zhu, ``Invariant
  attribute profiles: A spatial-frequency joint feature extractor for
  hyperspectral image classification,'' \emph{IEEE Trans. Geosci. Remote
  Sens.}, 2020, dOI:10.1109/TGRS.2019.2957251.

\bibitem{wohlberg2016convolutional}
B.~Wohlberg, ``Convolutional sparse representations as an image model for
  impulse noise restoration.'' in \emph{IVMSP}, 2016, pp. 1--5.

\bibitem{pinel2012multi}
B.~Pinel-Puyss{\'e}gur, R.~Michel, and J.-P. Avouac, ``Multi-link insar time
  series: Enhancement of a wrapped interferometric database,'' \emph{IEEE J.
  Sel. Topics Appl. Earth Observ. Remote Sens.}, vol.~5, no.~3, pp. 784--794,
  2012.

\end{thebibliography}
\begin{IEEEbiography}[{\includegraphics[width=1in,height=1.25in,clip,keepaspectratio]{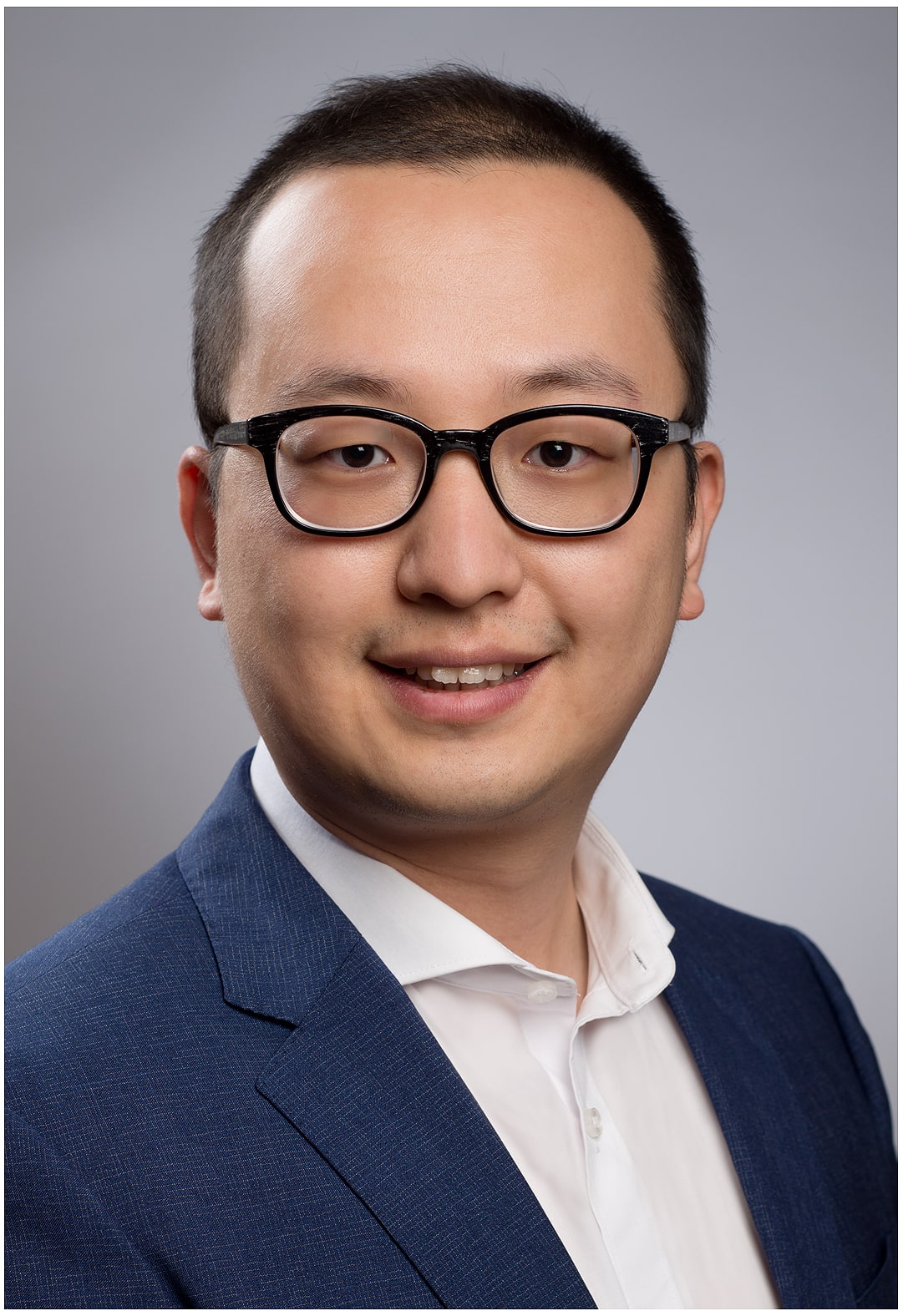}}]{Jian Kang} (S'16-M'19) received B.S. and M.E. degrees in electronic engineering from Harbin Institute of Technology (HIT), Harbin, China, in 2013 and 2015, respectively, and Dr.-Ing. degree from Signal Processing in Earth Observation (SiPEO), Technical University of Munich (TUM), Munich, Germany, in 2019. In August of 2018, he was a guest researcher at Institute of Computer Graphics and Vision (ICG), TU Graz, Graz, Austria. He is currently with Faculty of Electrical Engineering and Computer Science, Technische Universit\"at Berlin (TU-Berlin), Berlin, Germany. His research focuses on signal processing and machine learning, and their applications in remote sensing. In particular, he is interested in multi-dimensional data analysis, geophysical parameter estimation based on InSAR data, SAR denoising and deep learning based techniques for remote sensing image analysis. He obtained first place of the best student paper award in EUSAR 2018, Aachen, Germany.
\end{IEEEbiography}
\vskip -2\baselineskip plus -1fil
\begin{IEEEbiography}[{\includegraphics[width=1in,height=1.25in,clip,keepaspectratio]{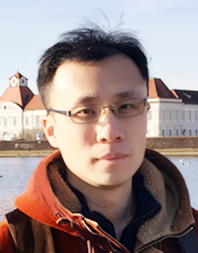}}]{Danfeng Hong}
(S'16-M'19) received the B.Sc. degree in computer science and technology, Neusoft College of Information, Northeastern University, China, in 2012, the M.Sc. degree (summa cum laude) in computer vision, College of Information Engineering, University, China, in 2015, the Dr. -Ing degree (summa cum laude) in Signal Processing in Earth Observation (SiPEO), Technical University of Munich (TUM), Munich, Germany, in 2019.

Since 2015, he worked as a Research Associate at the Remote Sensing Technology Institute (IMF), German Aerospace Center (DLR), Oberpfaffenhofen, Germany. Currently, he is a research scientist and leads a Spectral Vision group at EO Data Science, IMF, DLR, and also a permanent visiting scientist in GIPSA-lab, Grenoble INP, CNRS, Univ. Grenoble Alpes, Grenoble, France. In 2018 and 2019, he was a visiting scholar in GIPSA-lab, Grenoble INP, CNRS, Univ. Grenoble Alpes, Grenoble, France, and in RIKEN Artificial Intelligent Project (AIP), RIKEN, Tokyo, Japan. His research interests include signal / image processing and analysis, pattern recognition, machine / deep learning and their applications in Earth Vision.
\end{IEEEbiography}
\vskip -2\baselineskip plus -1fil
\begin{IEEEbiography}[{\includegraphics[width=1in,height=1.25in,clip,keepaspectratio]{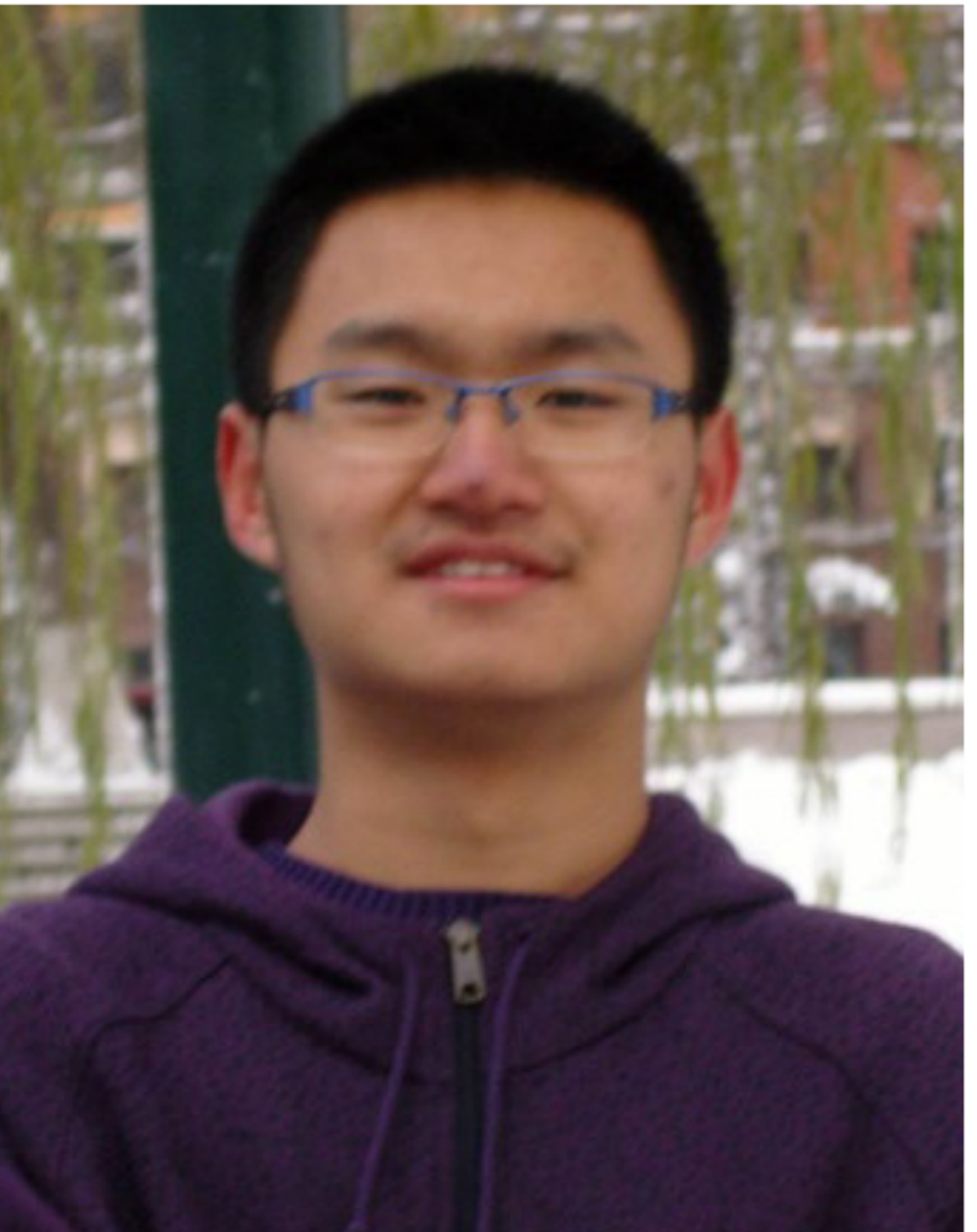}}]{Jialin Liu} received B.S. degree in automation and control from Tsinghua University, Beijing, China, in 2015 and is currently pursuing the Ph.D. degree in applied mathematics at University of California, Los Angeles (UCLA), Los Angeles, USA. His current research interest lies in the machine learning and optimization algorithms with applications in signal/image processing. He won ``Best Student Paper: Third Place'' at the 2017 International Conference on Image Processing (ICIP).
\end{IEEEbiography}
\vskip -2\baselineskip plus -1fil
\begin{IEEEbiography}[{\includegraphics[width=1in,height=1.25in,clip,keepaspectratio]{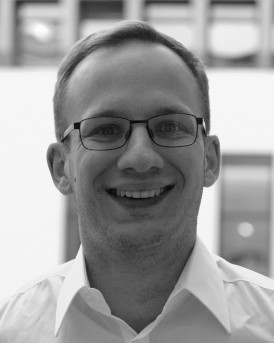}}]{Gerald Baier} received the M.Sc. degree in electrical engineering from the Université catholique de Louvain and the Karlsruhe Institute of Technology in 2012 and the Ph.D. degree from the Technical University of Munich in the field of synthetic aperture radar interferometry in 2018. From 2014 to 2018, he was at the department for SAR signal processing of the German Aerospace Center's (DLR) Remote Sensing Technology Institute.
Since 2018, he is a postdoctoral researcher in the Geoinformatics Unit at the RIKEN Center for Advanced Intelligence Project (AIP), Tokyo, Japan. His research interests include signal processing, machine learning, synthetic aperture radar and high-performance computing.
\end{IEEEbiography}
\vskip -2\baselineskip plus -1fil
\begin{IEEEbiography}[{\includegraphics[width=1in,height=1.25in,clip,keepaspectratio]{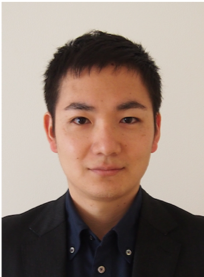}}]{Naoto Yokoya}
Naoto Yokoya (S'10-M'13) received the M.Eng. and Ph.D. degrees in aerospace engineering from the University of Tokyo, Tokyo, Japan, in 2010 and 2013, respectively.

He is currently a Unit Leader at the RIKEN Center for Advanced Intelligence Project, Tokyo, Japan, where he leads the Geoinformatics Unit since 2018. He is also a visiting Associate Professor at Tokyo University of Agriculture and Technology since 2019. He was an Assistant Professor at the University of Tokyo from 2013 to 2017. In 2015-2017, he was an Alexander von Humboldt Fellow, working at the German Aerospace Center (DLR), Oberpfaffenhofen, and Technical University of Munich (TUM), Munich, Germany. His research is focused on the development of image processing, data fusion, and machine learning algorithms for understanding remote sensing images, with applications to disaster management.

Dr. Yokoya won the first place in the 2017 IEEE Geoscience and Remote Sensing Society (GRSS) Data Fusion Contest organized by the Image Analysis and Data Fusion Technical Committee (IADF TC). He is the Chair (2019-2021) and was a Co-Chair (2017-2019) of IEEE GRSS IADF TC and also the secretary of IEEE GRSS All Japan Joint Chapter since 2018. He is an Associate Editor for the IEEE Journal of Selected Topics in Applied Earth Observations and Remote Sensing (JSTARS) since 2018. He is/was a Guest Editor for the IEEE JSTARS in 2015-2016, for Remote Sensing in 2016-2019, and for the IEEE Geoscience and Remote Sensing Letters (GRSL) in 2018-2019.
\end{IEEEbiography}
\vskip -2\baselineskip plus -1fil
\begin{IEEEbiography}[{\includegraphics[width=1in,height=1.25in,clip,keepaspectratio]{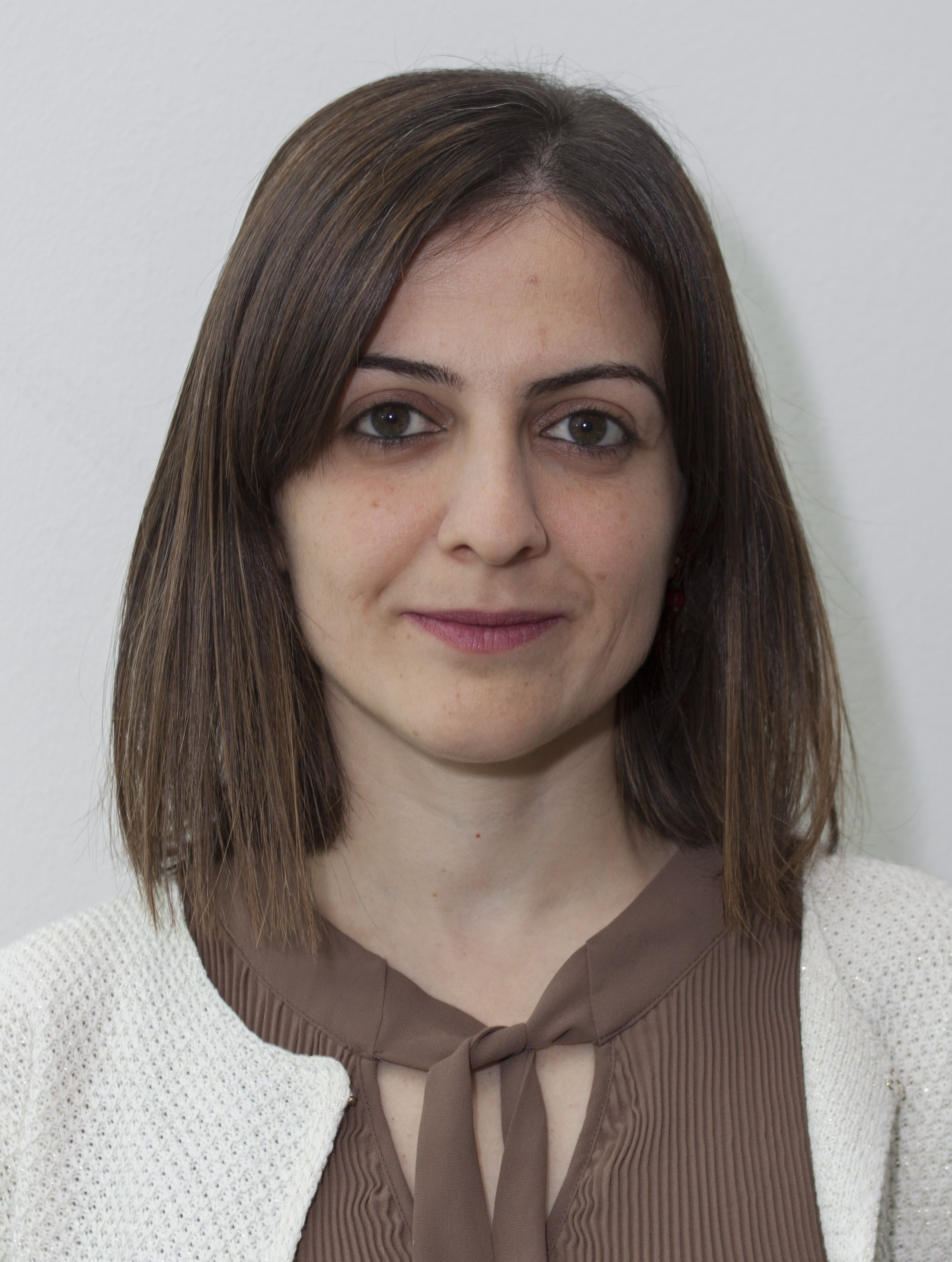}}]{Beg\"um Demir} (S’06–M’11–SM’16) received the B.S., M.Sc., and Ph.D. degrees in electronic and telecommunication engineering from Kocaeli University, Kocaeli, Turkey, in 2005, 2007, and 2010, respectively.

She is currently a Full Professor and head of the Remote Sensing Image Analysis (RSiM) group at the Faculty of Electrical Engineering and Computer Science, Technische Universit\"at Berlin, Germany since 2018. Before starting at TU Berlin, she was an Associate Professor at the Department of Computer Science and Information Engineering, University of Trento, Italy. Her research activities lie at the intersection of machine learning, remote sensing and signal processing. Specifically, she performs research on developing innovative methods for addressing a wide range of scientific problems in the area of remote sensing for Earth observation. She was a recipient of a Starting Grant from the European Research Council (ERC) with the project “BigEarth-Accurate and Scalable Processing of Big Data in Earth Observation” in 2017, and the “2018 Early Career Award” presented by the IEEE Geoscience and Remote Sensing Society. Dr. Demir is a senior member of IEEE since 2016.

Dr. Demir is a Scientific Committee member of several international conferences and workshops, such as: Conference on Content-Based Multimedia Indexing, Conference on Big Data from Space, Living Planet Symposium, International Joint Urban Remote Sensing Event, SPIE International Conference on Signal and Image Processing for Remote Sensing, Machine Learning for Earth Observation Workshop organized within the ECML/PKDD. She is a referee for several journals such as the PROCEEDINGS OF THE IEEE, the IEEE TRANSACTIONS ON GEOSCIENCE AND REMOTE SENSING, the IEEE GEOSCIENCE AND REMOTE SENSING LETTERS, the IEEE TRANSACTIONS ON IMAGE PROCESSING, Pattern Recognition, the IEEE TRANSACTIONS ON CIRCUITS AND SYSTEMS FOR VIDEO TECHNOLOGY, the IEEE JOURNAL OF SELECTED TOPICS IN SIGNAL PROCESSING, the International Journal of Remote Sensing, and several international conferences. Currently she is an Associate Editor for the IEEE GEOSCIENCE AND REMOTE SENSING LETTERS and Remote Sensing.
\end{IEEEbiography}
\end{document}